\documentclass[11pt]{article}
\usepackage{amsfonts}
\usepackage{indentfirst}
\usepackage{amsmath}
\usepackage{subcaption}
\usepackage{latexsym,amsmath,amssymb,paralist,graphicx}
\usepackage{enumerate}
\usepackage{multirow}
\usepackage{comment}
\usepackage[compress]{cite}
\usepackage{algorithm}
\usepackage{algorithmic}
\usepackage[title]{appendix}
\usepackage{xcolor}
\usepackage{graphicx}
\usepackage{epstopdf}
\usepackage{bm}
\usepackage{hyperref}
\usepackage[normalem]{ulem}
\usepackage{bbm}
\usepackage{mathrsfs}
\usepackage{booktabs}
\usepackage{slashed}
\usepackage{authblk}
\numberwithin{equation}{section}
\newtheorem{theorem}{Theorem}[section]
\newtheorem{lemma}{Lemma}[section]
\newtheorem{assumption}{Assumption}[section]
\newtheorem{proposition}{Proposition}[section]

\newtheorem{remark}{Remark}[section]
\newtheorem{definition}{Definition}[section]
\newenvironment{proof}{{\noindent\it Proof:\quad}}{\hfill$\square$\par}
\newcommand{\qed}{\nobreak \ifvmode \relax \else
      \ifdim\lastskip<1.5em \hskip-\lastskip
      \hskip1.5em plus0em minus0.5em \fi \nobreak
      \vrule height0.75em width0.5em depth0.25em\fi}

\def\XXint#1#2#3{{\setbox0=\hbox{$#1{#2#3}{\int}$}
\vcenter{\hbox{$#2#3$}}\kern-.51\wd0}}

\newenvironment{proofsketch}
{\noindent\textbf{Proof sketch.}}
{\hfill$\square$\par}
\begin{document}

\title{Responsive Distribution of G-normal Random Variables}

\author[1]{Ziting Pei}
\author[2]{Shige Peng}
\author[3,4]{Xingye Yue}
\author[3]{Xiaotao Zheng\thanks{Corresponding author. Email: xiaotaozheng99@gmail.com.}}

\affil[1]{School of Business, Suzhou University of Science and Technology, Suzhou 215006, Jiangsu, China}
\affil[2]{School of Mathematics, Shandong University, Jinan 250100, Shandong, China.}
\affil[3]{Center for Financial Engineering, Soochow University, Soochow University, Suzhou 215006, Jiangsu, China}
\affil[4]{School of Mathematical Sciences, Soochow University, Suzhou 215006, China}
\date{}

\maketitle

\begin{abstract}

A $G$-normal random variable $X\sim \mathcal{N}(0,[\underline{\sigma}^2,\overline{\sigma}^2])$ does not admit a unique probability law due to volatility uncertainty. For a given test function $\phi$, the $G$-expectation admits the stochastic control representation$$\mathbb{E}[\phi(X)]
=
\sup_{\sigma\in[\underline{\sigma},\overline{\sigma}]}
{E}\!\left[\phi(X_T^\sigma)\mid X_0^\sigma=0\right]
={E}\!\left[\phi(X_T^\ast)\mid X_0^\ast=0\right].$$
This formulation interprets the nonlinear expectation as a linear expectation under the law induced by the optimally controlled diffusion $X^\ast$, namely, the terminal law of $X_T^\ast$. This observation motivates the notion of a \emph{responsive distribution}, a measurement-dependent probability density $f_\phi$ such that, for a given test function $\phi$, 
$$\mathbb{E}[\phi(X)] = \int_{\mathbb{R}} \phi(x)\,f_\phi(x)\,dx.$$
Based on this viewpoint, we propose a coupled backward--forward trinomial tree framework for computing the $G$-expectation and constructing the corresponding responsive distribution. The backward trinomial tree discretizes the associated stochastic optimal control problem and yields approximations of the value function (i.e., the 
$G$-expectation) and the optimal feedback control, while the forward trinomial tree propagates the induced transition probabilities and produces a discrete approximation of the responsive distribution.

We establish rigorous convergence results for both components of the method. Using the equivalence between the backward trinomial tree and a monotone explicit finite difference scheme for the associated backward $G$-heat equation, we prove the convergence of the discrete value function to the value function of the stochastic control problem, and consequently, to the $G$-expectation.
More importantly, we show that the discrete second-order difference quotients generated by the backward trinomial tree converge locally uniformly to the second spatial derivative of the solution to the $G$-heat equation, $u_{xx}$. This yields the strong convergence of the discrete optimal control. Building on this result, we prove that the discrete probability measures generated by the forward recursion converge weakly$^\ast$ to the responsive distribution of the optimally controlled diffusion, equivalently, to a weak solution of the corresponding nonlinear Fokker--Planck equation.
Numerical results not only validate the theoretical convergence of the coupled schemes but also provide a powerful, practical sampling tool to visualize the complex responsive distributions under various measurements.

\end{abstract}

\section{Introduction}

Since the seminal work of Peng \cite{p06}, the theory of sublinear expectations, now known as $G$-expectation theory, has provided a systematic framework for probability theory under model uncertainty. 
In stark contrast to the classical Kolmogorov axiomatic framework, which relies on a single probability measure, the $G$-expectation (denoted by $\mathbb{E}[\cdot]$ throughout this paper) 
is a nonlinear functional that intrinsically incorporates ambiguity in the underlying probabilistic model, most notably in the form of uncertain volatility. 
From a risk-theoretic perspective, $G$-expectation induces a coherent risk measure in the sense of Artzner et al.~\cite{AD}.
Consequently, it has catalyzed significant advancements across various disciplines, including mathematical finance \cite{pei, p11, grigorova2023stochastic, pun2021g, faidi2025optimal}, economic theory \cite{fei2025influence, peng2022distributional}, and stochastic optimal control \cite{p10, li2025dynamic, sun2016maximum, li2025relationship, ren2018stabilization}.

A central object in $G$-expectation theory is the \emph{$G$-normal distribution}, which plays the role of the Gaussian distribution in the nonlinear framework. 
A random variable $X$ is defined to be $G$-normally distributed if it satisfies the property
\[
aX + bX' \stackrel{d}{=} \sqrt{a^2+b^2}\, X,
\quad a,b \ge 0,
\]
where $X'$ is an independent copy of $X$ under the $G$-expectation \cite{p06}. Unlike the classical normal distribution, a $G$-normal random variable does not admit a single
probability law. Instead, it is characterized by a nonlinearity
\[
\frac12\mathbb{E}[\left\langle {aX, X} \right\rangle ]=  G(a) = \frac12 \sup_{\sigma^2\in[\underline{\sigma}^2,\overline{\sigma}^2]} \sigma^2 a,
\]
reflecting uncertainty in the variance. The $G$-normal distribution arises naturally as the limit law in Peng’s nonlinear central limit theorem \cite{p10, peng2019law}, revealing a fundamental structural difference between linear probability spaces and sublinear expectation spaces. Understanding the expectation, distributional structure, and sampling mechanisms of $G$-normal random variables is therefore of both theoretical and practical importance.

In classical probability theory, linear expectations and probability distributions are typically approximated simultaneously via Monte Carlo sampling. Such methods are attractive because they simultaneously provide approximations of expectations and the underlying distribution. However, this paradigm collapses under $G$-expectation. A $G$-normal random variable encapsulates infinitely many linear distributions corresponding to different volatility scenarios, and naive sampling across these models does not yield a meaningful notion of distribution. As a consequence, even the basic task of computing $\mathbb{E}[\phi(X)]$ or identifying the distributional features of $X$ becomes highly nontrivial. A powerful analytical tool for addressing this problem is the nonlinear Feynman--Kac formula
under $G$-expectation \cite{p10}, which connects the $G$-expectation of a test function $\phi$ to a fully nonlinear parabolic equation.
More precisely, for a one-dimensional $G$-normal random variable $X$ and a bounded continuous function $\phi$, the function
\[
u(t,x) := \mathbb{E}[\phi(x+\sqrt{t}\,X)]
\]
is the viscosity solution of the $G$-heat equation
\begin{equation}\label{1dimen}
\partial_t u - \sup_{\sigma^2\in[\underline{\sigma}^2,\overline{\sigma}^2]}
\frac12 \sigma^2 \partial_{xx} u = 0,
\quad u(0,x)=\phi(x).
\end{equation}
When $\phi$ is convex or concave, the $G$-heat flow preserves this convexity or concavity, and the equation reduces to a linear heat equation with constant coefficients. 
In this case, the $G$-expectation collapses to a linear expectation, and the induced distribution is Gaussian with extremal variance \cite{p10}. These special
cases provide valuable intuition but cover only a limited class of measurements.

For general test functions that are neither convex nor concave, no analytic solution is available, and one must rely on numerical methods for the G-heat equations. In the one-dimensional setting, finite difference schemes and their convergence properties have been extensively studied by Pooley et al.~\cite{pooley} and Zhao et al.~\cite{yang2016numerical}. In the two-dimensional case, numerical treatment becomes more challenging due to cross-derivative terms and the need to preserve monotonicity. When the covariance structure is deterministic, wide-stencil explicit schemes were studied by Bonnans and Zidani \cite{Bonnans2003} and Debrabant and Jakobsen \cite{Debrabant2013}, while fully implicit methods were proposed by Ma and Forsyth \cite{MF}. When the sign of the correlation is uncertain, implicit schemes and convergence results were developed by Pei et al \cite{pei2025numerical}. 
Despite these advancements, existing literature predominantly treats the $G$-heat equation purely as a deterministic PDE to be solved. This perspective completely obscures the underlying probabilistic intuition, leaving the computation and interpretation of the \emph{distribution} of a $G$-normal random variable as an open problem. Specifically, there exists no unified numerical framework capable of simultaneously evaluating the $G$-expectation and providing a meaningful sampling strategy for $G$-normal random variables.

To bridge this gap, we adopt a stochastic control perspective.
Consider a one-dimensional $G$-normal random variable $X$ with volatility uncertainty $[\underline\sigma,\overline\sigma]$ and a bounded continuous test function $\phi$.
By reversing the time variable, the $G$-heat equation associated with $\mathbb{E}[\phi(X)]$ can be formulated in the backward form
\begin{equation}
	\left\{
	\begin{array}{l}
		\partial_{t} u+\frac{\sigma ^{2}\left( \partial_{xx} u\right) }{2}\partial_{xx} u=0,(t,x)\in
		(0,1)\times \mathbb{R}, \\
		u(t,x)|_{t=1}=\phi (x),%
	\end{array}%
	\right.  \label{bp}
\end{equation}%
where the nonlinearity is governed by
\begin{equation}
	\sigma ^{2}\left(\partial_{xx} u\right) =\left\{
	\begin{array}{l}
		\overline{\sigma }^{2},\text{ if }\partial_{xx} u\geq 0, \\
		\underline{\sigma }^{2},\text{ if }\partial_{xx} u<0.%
	\end{array}%
	\right.
\end{equation}
Crucially, this equation admits a stochastic control interpretation.
The solution \( u(t,x) \), with \( u(0,0) = \mathbb{E}[\phi(X)] \), coincides with the value function of a stochastic optimal control problem (SOCM) wherein the volatility \( \sigma \) is dynamically selected from \( [\underline{\sigma}, \overline{\sigma}] \) to maximize the expected terminal payoff.
More precisely, let $X_t^\sigma$ satisfy
\begin{equation}\label{primary_state_eq}
	dX_t^\sigma =\sigma(t, X_t^\sigma) dW_t, \quad X_0^\sigma=x_0.
\end{equation}
where $\sigma(\cdot,\cdot)$ is an admissible control taking values in $[\underline\sigma,\overline\sigma]$.
The associated value function is given by
\begin{equation}\label{primary_value_f}
	u(0, x_0=0)=\sup _{\sigma\in \left[\underline {\sigma},  \overline  {\sigma}\right]} E\left[\phi \left(X_T^\sigma\right) \mid X_0^\sigma=x_0 \right].
\end{equation}
The \( G \)-heat equation \eqref{bp} serves as the Hamilton--Jacobi--Bellman (HJB) equation for this SOCM; consequently, the (optimal) value function of the SOCM coincides with the viscosity solution of the G-heat equation. By  the verification theorem, the optimal feedback control $\sigma$ exhibits a bang–bang structure:
\begin{equation}\label{optimal_control_1}
	 \sigma^*(t,x)  = \left\{ \begin{array}{l}
		\overline  {\sigma}, \partial_{{x}{x}}u(t,x) \ge 0\\
		\underline {\sigma}, \partial_{{x}{x}}u(t,x) < 0
	\end{array} \right. ,
\end{equation}
where $\partial_{xx} u(t, x)$ denotes the second partial derivative of the value function $u(t, x)$ with respect to the state variable $x$.
Moreover, the corresponding optimally controlled process $\left({X}_t^*, t \geq 0\right)$ follows the SDE,
\begin{equation}\label{primary_optimal_state}
	d {X}_t^*=\sigma^*\left(t, {X}_t^*\right) d W_t, \quad {X}_0^*=0,
\end{equation}
provided that it is well-posed (i.e., it has a unique weak solution). 
The associated (optimal) value function $u(t, x)$ is given by:
\begin{equation}\label{primary_value_f_1}
	\begin{aligned}
		\mathbb{E}[\phi(X)] =	u(0, 0)&=\sup _{\sigma\in \left[\underline {\sigma},  \overline  {\sigma}\right]} E\left[\phi \left(X_T^\sigma\right) \mid X_0^\sigma=0 \right]\\
		&= E\left[\phi \left({X}_T^*\right) \mid {X}_0^*=0 \right] = \int_{\mathbb{R}} \phi(x)\,f_\phi(dx).
	\end{aligned}
\end{equation}
For a comprehensive treatment of classical stochastic control theory, the reader is referred to Yong and Zhou (1999) \cite{Yong} and Fleming (2006) \cite{Fleming06}. 
This control-theoretic representation provides a crucial insight: once the test function $\phi$ is fixed, the nonlinear $G$-expectation can be realized as a linear expectation under the terminal law of the optimally controlled process. 
Motivated by this observation, this paper introduces the concept of a \emph{responsive distribution} \( f_{\phi}(x) \) associated with a \( G \)-normal random variable \( X \). 
\begin{definition}[Responsive distribution]
Let \(X\) be a one-dimensional \(G\)-normal random variable and \(\phi:\mathbb{R}\to\mathbb{R}\) be a bounded continuous function. 
Denote by \(X_t^*\) the optimally controlled process associated with \(\phi\) as characterized by \eqref{primary_value_f_1}, i.e.,
\[
\mathbb{E}[\phi(X)] 
= \sup_{\sigma\in[\underline{\sigma},\overline{\sigma}]} 
E\!\left[\phi\!\left(X_T^\sigma\right)\mid X_0^\sigma=0\right]
= E\!\left[\phi\!\left(X_T^*\right)\mid X_0^*=0\right].
\]
The \emph{responsive distribution} associated with $\phi$, denoted by $\mu_\phi$, is defined as the probability law of the terminal state \(X_T^*\), such that
\[
\mathbb{E}[\phi(X)] = \int_{\mathbb{R}} \phi(x)\, \mu_\phi(dx).
\]
If \(\mu_\phi\) is absolutely continuous with respect to the Lebesgue measure, we denote its density by \(f_\phi(x)\), and write
\[
\mathbb{E}[\phi(X)] = \int_{\mathbb{R}} \phi(x)\, f_\phi(x)\,dx.
\]
\end{definition}
In this sense, although a $G$-normal random variable does not admit a universal probability law, it can be associated with a unique measurement-dependent probability distribution once the test function $\phi$ is specified. The responsive distribution thus provides a linear expectation representation of the nonlinear $G$-expectation, which is intrinsically linked to the underlying stochastic control mechanism.

Motivated by this viewpoint and by the Markov chain approximation framework of Kushner and Dupuis~\cite{Kushner, KushnerDupuis2013}, we develop two \emph{coupled trinomial tree methods} for both the numerical evaluation of the value function and the construction of the responsive distribution. 
The backward trinomial tree provides a monotone and stable numerical approximation of the value function $u(0,0)$, corresponding to the viscosity solution of the backward $G$-heat equation.
Based on the optimal control and transition structure identified by this backward recursion, a forward trinomial tree is constructed to propagate state-dependent transition probabilities, thereby yielding a discrete approximation of the responsive distribution \( f_{\phi}(x) \).
The two procedures are intrinsically linked: the backward recursion determines the optimal control and transition structure, while the forward recursion translates this structure into a discrete probability distribution.
Rigorous convergence of the proposed schemes are established in two stages. 
First, by interpreting the trinomial tree recursion as an explicit finite difference scheme, we prove that the associated discrete value function converges uniformly to the unique viscosity solution of the $G$-heat equation, within the Barles–Souganidis framework.
To overcome this, we derive from the backward trinomial tree a monotone finite difference scheme for an auxiliary variable associated with the discrete second-order difference quotients. This auxiliary variable serves as the discrete counterpart to $w := \sigma^2(u_{xx})u_{xx}$ , which satisfies a fully nonlinear partial differential equation at the continuous level. By proving the convergence of this auxiliary scheme, we recover the local uniform convergence of the discrete second-order difference quotients to $u_{xx}$.
This result yields the strong convergence of the discrete optimal control. Building on this convergence, we prove that the discrete probability measures generated by the forward trinomial tree converge weakly$^\ast$ to a weak solution of a nonlinear Fokker--Planck equation. 
This limiting measure is the responsive distribution of the optimally controlled diffusion $X_t^\ast$.
Finally, we perform numerical experiments to illustrate the effectiveness of the proposed methods. These results not only visualize the complex, measurement-dependent densities $f_\phi$ of $G$-normal random variables under various test functions but also empirically confirm the convergence of the coupled schemes.


The remainder of this paper is organized as follows. 
Section \ref{sec:3} presents the coupled trinomial tree schemes for quantifying G-expectation and sampling G-normal random variables. 
Section \ref{sec:4} establishes the stability and convergence of the backward scheme. 
Section \ref{sec:5} analyzes the convergence of the responsive distribution, showing that the discrete measures generated by the forward trinomial recursion converge to a weak solution of the associated Fokker--Planck equation. 
Section \ref{sec:6} presents numerical experiments illustrating the performance and convergence of the proposed methods. 
Section \ref{sec:7} concludes the paper.

\section{Two Coupled Trinomial Tree Schemes}
{\label{sec:3}} 
In this section, we develop a discrete numerical methodology to quantify 
$G$-expectations and construct responsive distributions. The approach consists of two coupled trinomial tree schemes. Section \ref{sec:3.2} introduces a backward trinomial tree that discretizes the underlying stochastic control problem, yielding the optimal volatility strategy and the value function. Based on the transition structure determined by this backward recursion, Section \ref{sec:3.3} constructs a forward trinomial tree to propagate probability measures. This coupled structure naturally leads to a measurement-dependent representation of $G$-normal random variables.

\subsection{Backward trinomial tree method} \label{sec:3.2}

The controlled diffusion \(\{{X}_t^\sigma\}_{0 \leq t \leq T}\) is a continuous-time Markov process. 
To construct a trinomial tree approximation, we discretize the state space as $x_i = ih$, $i = 0, \pm1, \dots, \pm N$, and approximate \eqref{primary_state_eq} by a controlled Markov chain following \cite{Kushner}. 
Consequently, the controlled Markov chain \({X}_t^{h,\sigma}\) is expressed as:
\begin{equation}\label{DTMK}
	\begin{aligned}
		X_{t_{n+1}}^{h,\sigma} =  
		\begin{cases}
			x_i - h, & p = P_{i,i-1}^n(\sigma); \\
			x_i, & p = P_{i,i}^n(\sigma); \\
			x_i + h, & p = P_{i,i+1}^n(\sigma),
		\end{cases}
	\end{aligned}
\end{equation} 
conditional on the current state \(X_{t_{n}}^{h,\sigma} = x_i\). 
The transition probabilities for this controlled Markov chain are defined as:
\begin{equation}\label{transition_probabilities_DP}
	P_{i,i-1}^n(\sigma) = P_{i,i+1}^n(\sigma) = \frac{\sigma^2}{2} \frac{\Delta t}{h^2}, \quad P_{i,i}^n(\sigma) = 1 - P_{i,i-1}^n(\sigma) - P_{i,i+1}^n(\sigma).
\end{equation}
Here, \(P_{i,i-1}^n(\sigma)\), \(P_{i,i}^n(\sigma)\), and \(P_{i,i+1}^n(\sigma)\) represent the probabilities of transition from the current state \(X_{t_{n}}^{h,\sigma}\) to the states \(x_i - h\), \(x_i\), and \(x_i + h\), respectively. 
To guarantee that ransition probabilities are non-negative, we impose the standard stability condition:
\begin{equation}\label{eq:Cond_CFL}
\overline{\sigma}^{2} \frac{\Delta t}{h^{2}} \leq 1.
\end{equation}
Let $U\left( t_{n},x_{i}\right)$ or $U_{i}^n$ be the approximate of the value function for equation \eqref{bp} at $\left( t_{n},x_{i}\right)$.
By the dynamic programming principle, the discrete value function satisfies the following backward recursion:
\begin{equation} \label{TTM_1}
	\begin{aligned}
		U_{i}^{n} &= \sup _{\sigma\in \left[\underline {\sigma},  \overline  {\sigma}\right]} E\left[U \left(t_{n+1}, X_{t_{n+1}}^{h,\sigma}\right) \mid X_{t_n}^{h,\sigma} =x_i \right]\\
        & = \sup _{\sigma\in \left[\underline {\sigma},  \overline  {\sigma}\right]} \left(P_{i,i-1}^n(\sigma) U(t_{n+1},x_{i-1}) + P_{i,i}^n(\sigma) U(t_{n+1},x_i) + P_{i,i+1}^n(\sigma) U(t_{n+1},x_{i+1})\right)  \\
        & = \sup _{\sigma\in \left[\underline {\sigma},  \overline  {\sigma}\right]} 
        \frac{\sigma^2}{2} \Delta t \left( \frac{U_{i+1}^{n+1}-2U_{i}^{n+1}+U_{i-1}^{n+1}}{h^{2}} \right) + U_{i}^{n+1}\\
		& = P_{i,i-1}^n(\sigma_i^n) U_{i-1}^{n+1} + P_{i,i}^n(\sigma_i^n) U_{i}^{n+1} + P_{i,i+1}^n(\sigma_i^n) U_{i+1}^{n+1},
	\end{aligned} 
\end{equation}
where the discrete optimal feedback control $\sigma_i^n$ exhibits a bang-bang structure:
\begin{equation}\label{control_DP}
	\sigma_i^n = 
	\begin{cases}
		\overline{\sigma} & \text{if } \delta _{h}^{2}U_{i}^{n+1} > 0, \\
		\underline{\sigma} & \text{if } \delta _{h}^{2}U_{i}^{n+1}< 0,
	\end{cases}
\end{equation}
and $\delta_{h}^{2}$ denotes the standard central second-order difference operator:
\begin{equation}\label{delta_h^2}
	\delta _{h}^{2}U_{i}^{n+1} =\frac{U_{i+1}^{n+1}-2U_{i}^{n+1}+U_{i-1}^{n+1}}{h^{2}}.
\end{equation}
This backward recursion is initialized with the terminal condition $U_i^N = \phi(x_i)$ and proceeds backward in time for $n = N-1, \dots, 0$ over the spatial nodes $i = -n, \dots, n$.
Consequently, the backward trinomial tree \eqref{TTM_1}
provides a discrete realization of the stochastic control problem and yields a numerical approximation of the G-expectation: $U(0,0)\approx u(0,0)= \mathbb{E}[\phi (X)]$.

\begin{figure}[h]
	\centering
	\includegraphics[width=0.65\textwidth]{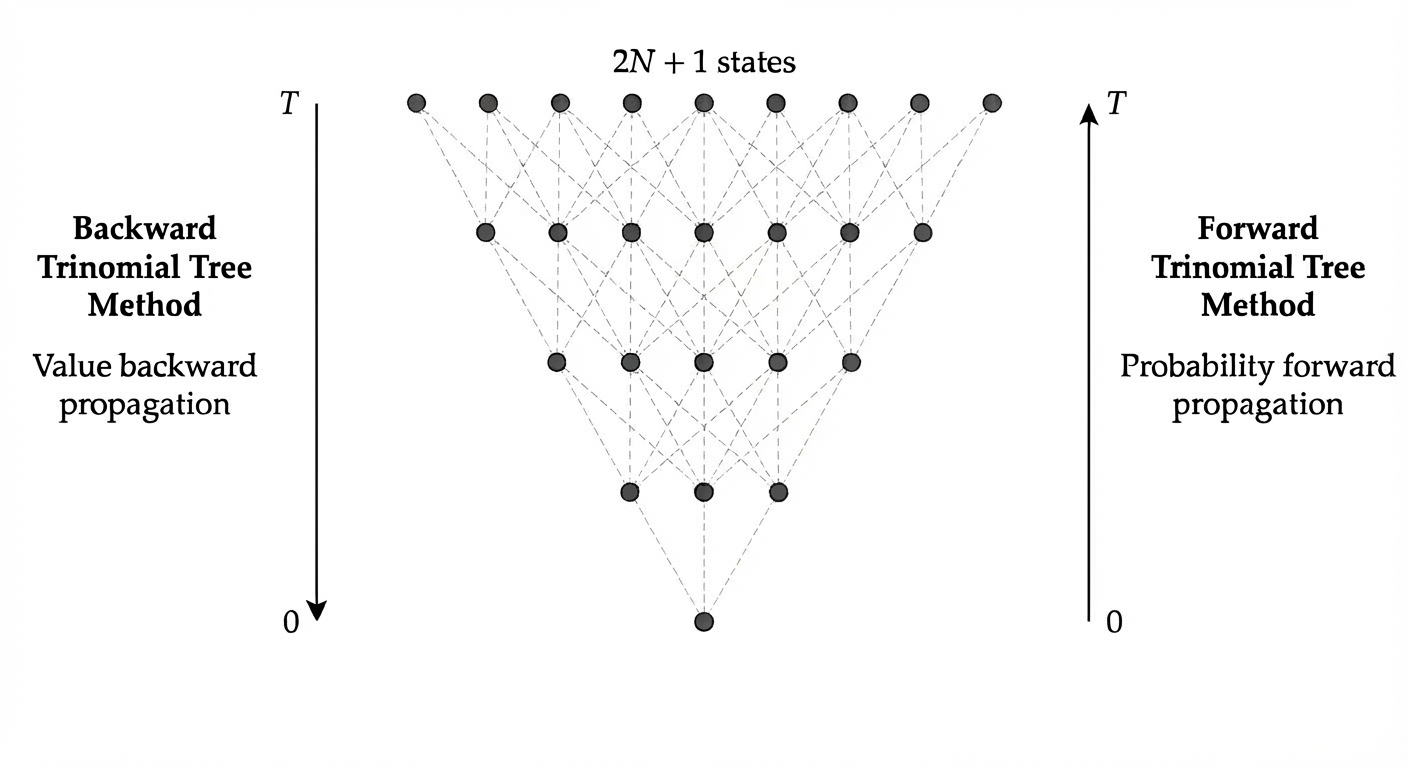}
	\caption{Trinomial tree with $2N+1$ spatial nodes. The backward recursion \eqref{TTM_1}
propagates the value function from $t_{n+1}$ to $t_{n}$, while the forward recursion \eqref{eq:sampling2_p} propagates the induced probability mass.}
	\label{fig:trinomial_tree_diagram}
\end{figure}

\begin{figure}[ht]
    \centering
    \begin{subfigure}[t]{0.45\textwidth}
        \centering
\includegraphics[width=\textwidth]{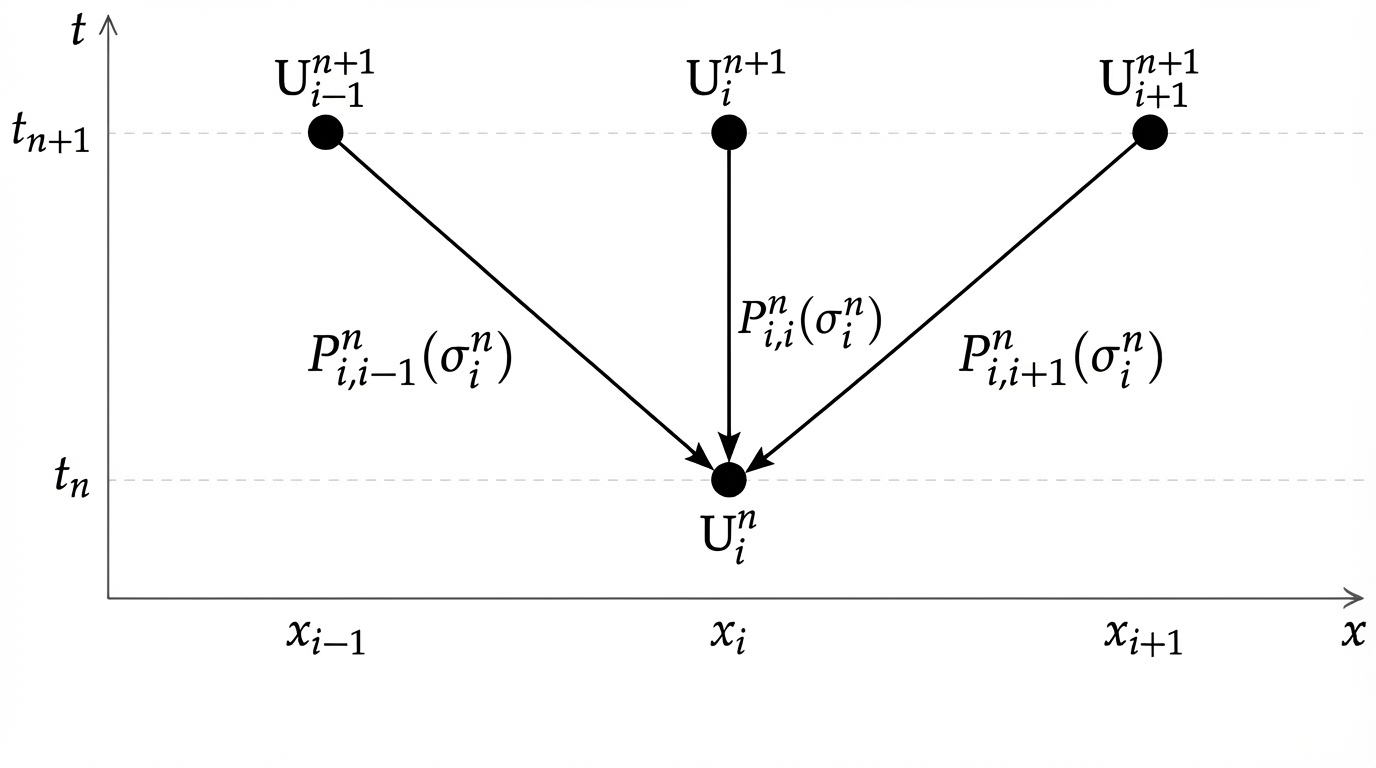}
        \caption{Backward trinomial tree method.}
        \label{fig:EFDM}
    \end{subfigure}
    \hfill
    \begin{subfigure}[t]{0.45\textwidth}
        \centering
\includegraphics[width=\textwidth]{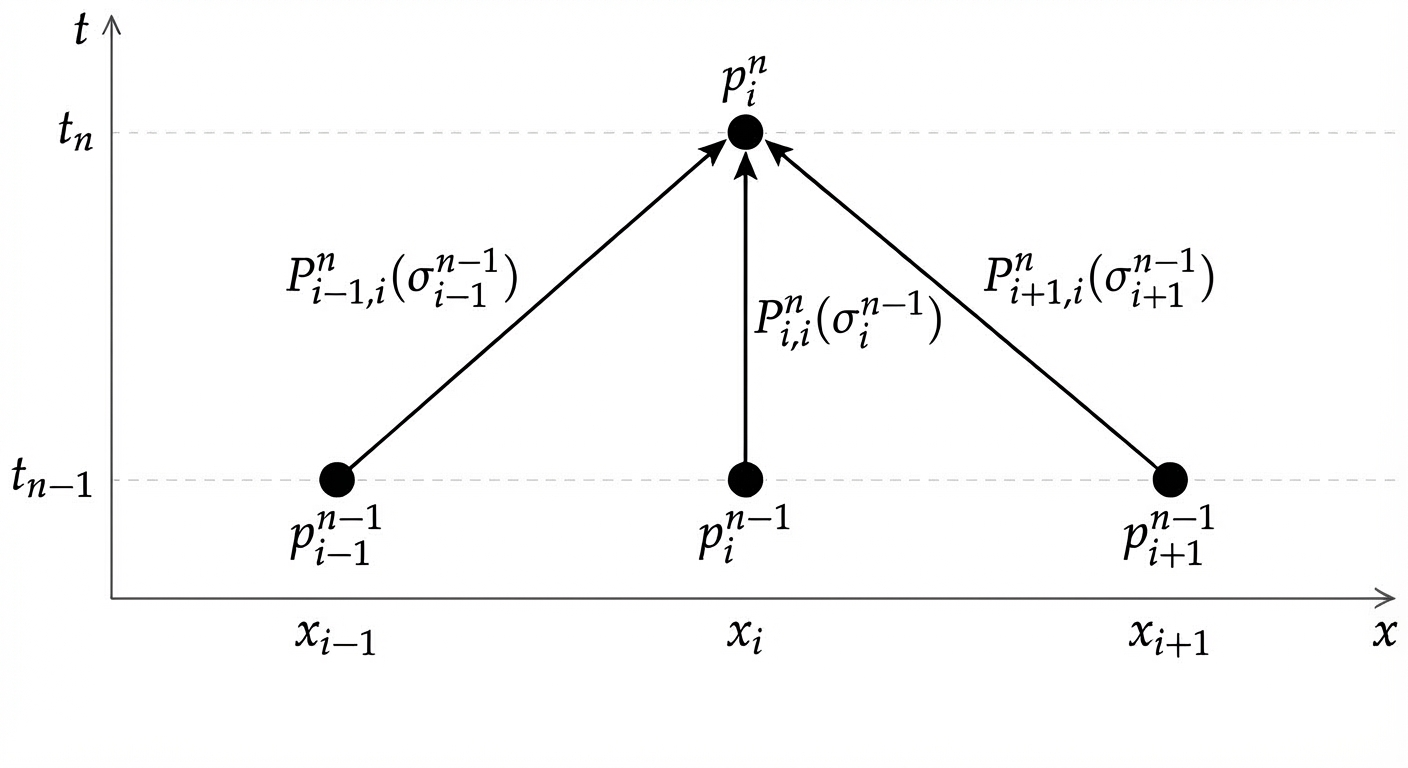}
        \caption{Forward trinomial tree method.}
        \label{treediagram1}
    \end{subfigure}
    \caption{Interpretation of trinomial tree methods.}
    \label{figure:Interpretation}
\end{figure}

\subsection{Sampling approach derived from forward trinomial tree method}\label{sec:3.3}


The backward recursion in Section~\ref{sec:3.2} determines not only the $G$-expectation but also an optimal feedback control policy $\{\sigma_i^n\}$. This policy dictates a specific, state-dependent transition structure. By utilizing this structure, we can propagate probability mass forward in time, thereby constructing a discrete approximation of the responsive distribution $f_\phi$.

Let $p_i^n$ denote the probability that the optimally controlled Markov chain
$X^{h,\sigma^\ast}_{t_n}$ occupies state $x_i$ at time $t_n$.
We take the discrete initial law
\[
p_0^0=1,\quad p_i^0=0\quad \text{for }i\neq0,
\]
corresponding to the initial condition $X_0=0$.
The probability mass then propagates forward according to
\begin{equation}
p_i^n
=
P^n_{i-1,i}(\sigma^{n-1}_{i-1})\,p^{n-1}_{i-1}
+
P^n_{i,i}(\sigma^{n-1}_{i})\,p^{n-1}_{i}
+
P^n_{i+1,i}(\sigma^{n-1}_{i+1})\,p^{n-1}_{i+1},
\qquad i=-n,\ldots,n,
\label{eq:sampling2_p}
\end{equation}
where the transition probability $P_{i,j}^n$ has been settled in the backward procedure \eqref{TTM_1}-\eqref{control_DP}.
The corresponding diagram of the forward sampling process is shown in Figure \ref{treediagram1}.
After $N$ time steps, the terminal distribution $\{p_i^N\}_{i=-N}^N$
defines a discrete probability measure that approximates the
responsive distribution $f_\phi$ introduced in Definition~3.1.
In particular,
\begin{equation}
\mathbb{E}[\phi(X)]
= \int_{\mathbb{R}} \phi(x)\,f_\phi(dx)
\;\approx\;
\sum_{i=-N}^N p_i^N\,\phi(x_i).
\label{eq:sampling_representation}
\end{equation}


It is important to emphasize that a G-normal random variable does not admit a universal probability law independent of the test function. Instead, each choice of measurement $\phi$ activates a specific optimal control, which in turn induces a unique probability measure $f_\phi$. As depicted in Figures \ref{fig:trinomial_tree_diagram} and \ref{fig:EFDM}, the forward trinomial tree provides a sampling procedure that is intrinsically measurement-dependent, reflecting the nonlinear nature of $G$-expectation.


\begin{remark}
Direct sampling of a G-normal random variable $X \sim \mathcal{N}(0, \left[\underline{\sigma }^{2}, \overline{\sigma}^2\right])$ is infeasible and not meaningful. For example, sampling across all possible linear normal random variables $X_\sigma \sim \mathcal{N}(0, \sigma^2)$ with $\sigma \in [\underline{\sigma}, \overline{\sigma}]$ does not provide useful information. 
However, once a specific measurement $\phi$ is introduced, sampling to evaluate $\mathbb{E}[\phi(X)]$ becomes both mathematically rigorous and practically feasible. 
This situation is comparable to quantum mechanics, where the state is uncertain and unknown until a measurement is made, at which point the system becomes well-defined. Therefore, one should only sample for purpose and never sample for the sake of sampling.
This paradigm dictates that under Knightian uncertainty, one must sample with respect to a specific objective function, rather than attempting to sample the underlying ambiguous state space in isolation.
\end{remark}


\section{Numerical analysis for Backward Trinomial Tree Method}
{\label{sec:4}}


In this section, we establish the convergence of the backward trinomial tree method for approximating the $G$-expectation. 
More precisely, we show that the numerical solution produced by the backward trinomial tree converges to the value function $u(0,0)$ of the associated stochastic control problem, which coincides with the viscosity solution of the $G$-heat equation and thus with $\mathbb{E}[\phi(X)]$.
The analysis proceeds in two steps. First, we interpret the trinomial tree recursion as a fully explicit finite difference scheme for the $G$-heat equation and establish its consistency, stability, and monotonicity. By the Barles--Souganidis framework, this implies convergence to the unique viscosity solution of the $G$-heat equation. 
Second, exploiting the equivalence between the trinomial tree method and the explicit scheme, we deduce the convergence and stability of the backward trinomial tree method.

\subsection{Equivalence with an Explicit Finite Difference Scheme}

The trinomial tree method admits an equivalent interpretation as an explicit finite difference discretization of the G-heat equation \eqref{bp}. This equivalence plays a crucial role in the convergence analysis. Indeed, the backward recursion \eqref{TTM_1} can be rewritten as
\begin{equation}
	\left\{
	\begin{array}{l}
		U_{i}^{n}=\frac{\sigma _{i}^{2}}{2}\frac{\Delta t}{h^{2}}U_{i-1}^{n+1}+%
		\left( 1-\sigma _{i}^{2}\frac{\Delta t}{h^{2}}\right) U_{i}^{n+1}+\frac{%
			\sigma _{i}^{2}}{2}\frac{\Delta t}{h^{2}}U_{i+1}^{n+1},\text{ }n=N-1,\cdots
		,1,0,\text{ }i=-n,\cdots ,n \\
		U_{i}^{N}=\phi (x_{i}),\text{ }%
	\end{array}%
	\right.  \label{dbp2}
\end{equation}%
where the optimal discrete volatility is given by
\begin{equation} \label{eq:sigma_optimal}
	\sigma _{i}^{2}=\sigma ^{2}(\delta_h^2 U_i^{n+1}) =\left\{
	\begin{array}{ll}
		\overline{\sigma }^{2}, & \text{if } \delta_h^2 U_i^{n+1} > 0, \\
		\underline{\sigma }^{2}, & \text{if } \delta_h^2 U_i^{n+1} \le 0.%
	\end{array}%
	\right.
\end{equation}%
Since the transition probabilities \(P_{i,i-1}(\sigma)\), \(P_{i,i}(\sigma)\), and \(P_{i,i+1}(\sigma)\) are nonnegative and sum to one, the scheme \eqref{dbp2} is a monotone probabilistic scheme. Rearranging \eqref{dbp2} yields the equivalent explicit finite difference form
\begin{equation}
	\left\{
	\begin{array}{l}
		\delta _{t}U_{i}^{n+1}+\frac{1}{2}\sigma ^{2}\left( \delta
		_{h}^{2}U_{i}^{n+1}\right) \delta _{h}^{2}U_{i}^{n+1}=0,\text{ }n=N-1,\cdots
		,1,0,\text{ }i\in\mathbb Z, \\
		U_{i}^{N}=\phi (x_{i}),\text{ }i=-N,\cdots ,N,%
	\end{array}%
	\right.  \label{dbp}
\end{equation}%
where the forward time difference is defined as
\begin{eqnarray*}
	\delta _{t}U_{i}^{n+1} &=&\frac{U_{i}^{n+1}-U_{i}^{n}}{\Delta t},
\end{eqnarray*}%
and the central second-order difference $\delta_h^2 U_i^{n+1}$ is given by \eqref{delta_h^2}.
It should be noted that the explicit scheme \eqref{dbp} is formally defined on the entire rectangular grid $(t_n,x_i)$ with $i\in\mathbb{Z}$, whereas the trinomial tree recursion \eqref{dbp2} only involves the triangular computational domain $i=-n,\dots,n$. 
However, the values outside this triangular region do not affect the domain of dependence for the root node $U_0^0$. In particular, the value depends only on the nodes contained in the triangular region of the trinomial tree. Therefore, the value $U_0^0$ obtained from the explicit scheme \eqref{dbp} is identical to that produced by the trinomial tree recursion.

\begin{figure}[h]
	\centering
	\includegraphics[width=0.6\textwidth]{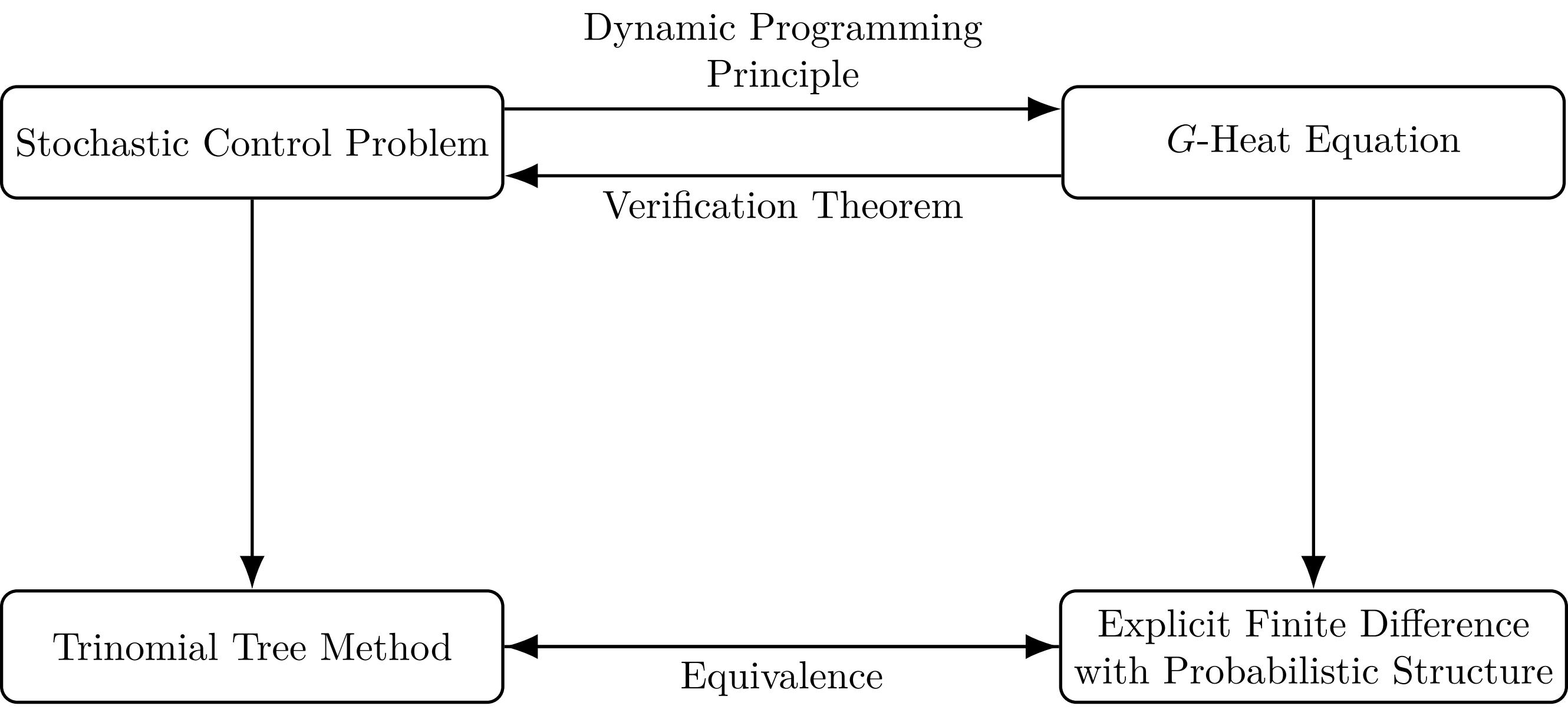}
	\caption{Swap Diagram between the stochastic control problem and G-heat equation.}
	\label{fig:swap_plot}
\end{figure}

\begin{remark}
The numerical equivalence between the trinomial tree method \eqref{TTM_1} and the explicit finite difference scheme \eqref{dbp2}, as depicted in Figure \ref{figure:Interpretation}, reflects a profound theoretical duality. As illustrated by the structural diagram in Figure \ref{fig:swap_plot}, this discrete-level equivalence stems naturally from the continuous-time relationship between the stochastic optimal control problem and its corresponding Hamilton--Jacobi--Bellman (HJB) equation, i.e., the $G$-heat equation.
\end{remark}

\subsection{Convergence of the Backward Trinomial Tree to the G-expectation}
Following the seminal framework of Barles and Souganidis \cite{G1991Convergence}, the numerical scheme \eqref{dbp} converges to the unique viscosity solution of the fully nonlinear PDE \eqref{bp} provided that the scheme is consistent, $l_{\infty }$-stable, and monotone.

We first recall the notion of monotonicity. Denote by $$g_{i} = g\left(U^{n}_{i}, U^{n+1}_{i}, \{ U^{n+1}_{k} \}_{k\in N_{i}}\right)$$ the left-hand side of the difference equation \eqref{dbp}. Here, $N_{i}=\{k\not = i :\ |k-i|\leq 1\}$ presents the set of all nearest-neighbor indexes of $i$.
\begin{definition} [Monotonicity] \label{def-mono}
    The scheme \eqref{dbp} is monotone if for all  $i$, 
    \begin{equation}
    \label{g-diag}
    g_{i}\left(U^{n}_{i}+\epsilon^{n}_{i}, U^{n+1}_{i}, \{U^{n+1}_{k}\}_{k\in N_{i}}\right)\leq  g_{i}\left(U^{n}_{i}, U^{n+1}_{i}, \{U^{n+1}_{k}\}_{k\in N_{i}}\right), \ \ \forall \epsilon^{n}_{i}\geq 0,
    \end{equation}
    and
    \begin{equation} 
    \label{g-off-diag}
      g_{i}\left(U^{n}_{i}, U^{n+1}_{i}+\epsilon^{n+1}_{i}, \{U^{n+1}_{k}+\epsilon^{n+1}_{k}\}_{k\in N_{i}}\right)\geq  g_{i}\left(U^{n}_{i}, U^{n+1}_{i}, \{U^{n+1}_{k}\}_{k\in N_{i}}\right), \  \forall \epsilon^{n+1}_{i}, \epsilon^{n+1}_{k}\geq 0.
    \end{equation}
 \end{definition}

\begin{assumption}[Stability condition]
\label{ass:CFL_w}
The time step $\Delta t$ and spatial step $h$ satisfy the mesh ratio constraint:
\begin{equation}\label{eq:CFL_w}
\overline{\sigma}^{2}\,\frac{\Delta t}{h^{2}}\le \frac{1}{2}.
\end{equation}
\end{assumption}

\begin{lemma}
\label{lemma:MSC}
Under Assumption \ref{ass:CFL_w}, 
the explicit scheme \eqref{dbp} satisfies the following properties:

\begin{enumerate}
\item[(i)] \textbf{Consistency:} the scheme is consistent with the G-heat equation \eqref{bp};
\item[(ii)] \textbf{Stability:} it is $l_\infty$-stable, i.e.,
\[
\max_n \|U^n\|_\infty \le \|\phi\|_\infty;
\]
\item[(iii)] \textbf{Monotonicity:} it is monotone in the sense of Definition \ref{def-mono}.
\end{enumerate}
\end{lemma}

The proof consists of direct verifications based on the 
probabilistic structure of the scheme and the stability condition 
\eqref{eq:CFL_w}. 
For completeness, all details are provided in Appendix~\ref{app_prf_lemma1}.

\begin{lemma}
[Convergence to the viscosity solution]
\label{thm:viscosity_convergence}
Let Assumption \ref{ass:CFL_w}
hold, then the fully explicit discretization %
\eqref{dbp} converges to the viscosity solution of the
non-linear PDE \eqref{bp}.
\end{lemma}

\begin{proof}
Since the  scheme \eqref{dbp} is consistent, $%
l_{\infty }-$stable, and monotone, the convergence follows from the results of Barles and Souganidis {\cite{G1991Convergence}} directly.
\end{proof}

\begin{lemma}[Rate of convergence, {\cite[Theorem 4.11]{pei2025numerical}}]
Let $u$ be the viscosity solution of
equation \eqref{bp}, and let $U$ be the numerical solution of equation
\eqref{dbp}. If Assumption \ref{ass:CFL_w} holds, $\exists $ $\beta \in
\left( 0,1\right) $, such that
\begin{equation}
\left\vert \left\vert u_{i,j}^{n}-U_{i,j}^{n}\right\vert \right\vert
_{\infty }\leq O\left( \Delta t^{\frac{\beta }{2}}+h^{\beta }\right) ,\text{
}\beta \in \left( 0,1\right) .
\end{equation}
\end{lemma}


The EFDM is equivalent to the Trinomial tree approach. we now state the main convergence result for the backward trinomial tree method (see Theorem \ref{Th_TTM_stable}).

\begin{theorem}\label{Th_TTM_stable}
Let $\Delta t = T/N$ and choose the spatial step $h$ such that the stability condition $\frac{\bar{\sigma}^2 \Delta t}{h^2} \leq 1$ is satisfied. Let $U_0^0$ be the value at the root node generated by the backward trinomial tree recursion \eqref{TTM_1}. Then, as $N \to \infty$, $U_0^0$ converges to the G-expectation with the terminal payoff $\phi(X)$:
\begin{equation}
    \lim_{N \to \infty} U_0^0 = u(0,0) = \mathbb{E}[\phi(X)],
\end{equation}
where $u(0,0)$ is the value function of the associated stochastic optimal control problem evaluated at the initial state.
\end{theorem}

\begin{proof}
By virtue of the algebraic equivalence established in Section 3.1, the root node value $U_0^0$ computed via the trinomial tree \eqref{TTM_1} is identical to the solution $U_0^0$ of the explicit finite difference scheme \eqref{dbp}. Under the specified stability condition \ref{ass:CFL_w}, Theorem \ref{thm:viscosity_convergence} guarantees that $U_0^0 \to u(0,0)$ as $\Delta t, h \to 0$. By the nonlinear Feynman-Kac formula under G-expectation, the viscosity solution of the G-heat equation at $(t,x)=(0,0)$ coincides exactly with the value function of the stochastic control problem, which is the definition of the G-expectation $\mathbb{E}[\phi(X)]$. This completes the proof.
\end{proof}


\section{Convergence of the Responsive Distribution} \label{sec:5}

The trinomial tree method proposed in Section \ref{sec:3} not only yields a numerical procedure for computing the nonlinear expectation $\mathbb{E}[\phi(X)]$, but also induces a family of discrete probability distributions $\{p^n_i\}$ associated with the controlled Markov chain representation of the optimal process. 
In this section, we prove that these discrete distributions converge weakly$^\ast$ to the responsive distribution of the optimally controlled diffusion process $X_t^\ast$. This limit is characterized as a weak solution of the corresponding nonlinear Fokker--Planck equation.

\subsection{G-normal Distribution and Fokker-Planck Equation}

We begin by recalling that, for a fixed test function $\phi$, the $G$-expectation admits the stochastic control representation
$$\mathbb{E}[\phi(X)] = \sup_{\sigma\in[\underline\sigma,\overline\sigma]} E\bigl[ \phi(X_T^\sigma) \mid X_0^\sigma = 0 \bigr]=
E\bigl[\phi(X_T^\ast) \mid X_0^\ast = 0  \bigr] =  \int_{\mathbb{R}} \phi(x)\,f_\phi(X)(dx),$$
where $X^\ast_t$ is the optimally controlled diffusion and $f_\phi$ denotes its terminal law, namely the responsive distribution associated with $\phi$. Therefore, once the measurement $\phi$ is fixed, the nonlinear expectation $\mathbb{E}[\phi(X)]$ can be represented as a linear expectation with respect to a measurement-dependent probability measure.


Let $p(t,x)$ denote the probability density of $ X_t^*$ at time $t$
and position $x$. Standard results from stochastic control and
Fokker--Planck theory (see, e.g., \cite{risken1989fokker,barbu2020nonlinear}) imply that $p$ satisfies
the forward Kolmogorov equation,
\begin{equation}\label{eq:p-Fokker-plank}
	\left\{
	\begin{array}{l}
   \partial_{t} p - \partial_{xx}\left( \frac{\sigma ^{2}(v)}{2} p \right)=0, \quad (t,x)\in
		(0,1]\times \mathbb{R}, \\
		p(t,x)|_{t=0}=\delta (x-0).%
	\end{array}%
	\right.  
\end{equation}
where $v = \partial_{xx} u$ is the second derivative of the solution to the G-heat equation \eqref{bp}, and the volatility coefficient
$\sigma^2(v)=\sigma^2(\partial_{xx} u)$ is determined by the bang--bang optimal control
rule \eqref{optimal_control_1}.

Although the second derivative $\partial_{xx} u$ exists, the composite coefficient
$\sigma^2(\partial_{xx} u)$ may fail to be regular due to the discontinuity of
$\sigma^2(\cdot)$. As a consequence, equation~\eqref{eq:p-Fokker-plank} is considered in the weak (formulation) sense.

\begin{definition}[Weak Solution]
A probability measure family ${p(t,\cdot)}_{t \geq 0}$ is called a weak solution of the Fokker-Planck equation \eqref{eq:p-Fokker-plank} if, for any $T > 0$ and any test function ${\varphi} \in C_c^{{\infty}}([0,T] \times \mathbb{R})$ with compact support and $\varphi(\cdot,T) = 0$, the following holds:
\begin{equation}\label{eq:weak_form}
\int_0^T \int_{{\mathbb{R}}} p\left( \partial_{t} \varphi+\frac{\sigma^2(\partial_{xx} u)}{2} \partial_{xx} \varphi\right) dx dt - \int_\mathbb{R} p(T,x)\varphi(T,x) dx + \int_\mathbb{R} p(0,x)\varphi(0,x) dx   = 0.
\end{equation}
\end{definition}

A natural first step is to investigate the evolution equation satisfied by $v = \partial_{xx} u$. Formally differentiating the $G$-heat equation \eqref{bp}
twice in space yields the nonlinear equation
\begin{equation}\label{eq:v-G}
	\begin{cases}
		\partial_t v + \dfrac{1}{2}\,\partial_{xx}\!\bigl(\sigma^2(v)\,v\bigr) = 0,
		& (t,x)\in(0,1)\times\mathbb{R},\\[0.4em]
		v(1,x) = \phi''(x),
	\end{cases}
\end{equation}
which must be understood in a weak sense.
However, due to the discontinuity of the coefficient $\sigma^2(\cdot)$,
there is no suitable viscosity solution framework available for
equation~\eqref{eq:v-G}, which makes it unsuitable for direct analytical
or numerical treatment.
To overcome this difficulty, we introduce the auxiliary variable,
\begin{equation}\label{eq:def-w}
w(t,x):=\sigma^2\!\bigl(\partial_{xx} u\bigr)\,\partial_{xx} u(t,x).
\end{equation}
Since the coefficient \(\sigma^2(\cdot)\) is piecewise constant, it is constant on the regions \(\{v>0\}\) and \(\{v<0\}\). Consequently, on each of these domains, we have
\[
    \partial_t\!\bigl(\sigma^2(v)v\bigr) = \sigma^2(v)\,\partial_t v.
\]
At points where $v=0$, one may formally write
$$\partial_t \bigl(\sigma^2(v)v\bigr) = \sigma^2(v)\partial_t v + (\overline{\sigma}^2 - \underline{\sigma}^2) v \delta(v) \partial_t v.$$
Since $v \delta(v) = 0$, the singular contribution to the derivative vanishes, ensuring that the identity remains valid in the distributional sense.
Therefore, the transformation \(w=\sigma^2(v)v\) is compatible with the time derivative. Multiplying equation~\eqref{eq:v-G} by \(\sigma^2(v)\) and using \(w=\sigma^2(v)v\) together with \(\sigma^2(v)=\sigma^2(w)\), we conclude that \(w\) satisfies the fully nonlinear parabolic equation
\begin{equation}\label{eq:w-equation}
\begin{cases}
\partial_t w + \dfrac12\,\sigma^2(w)\,\partial_{xx} w = 0,
& (t,x)\in(0,1)\times\mathbb{R},\\[0.4em]
w(1,x) = \sigma^2\!\big(\phi''(x)\big)\,\phi''(x),
& x\in\mathbb{R}.
\end{cases}
\end{equation}
Because the coefficient $\sigma^2(\cdot)$ is piecewise constant and exhibits a discontinuity at $w=0$, classical solutions to \eqref{eq:w-equation} generally do not exist. Consequently, the equation must be understood in the viscosity sense.
The auxiliary variable $w$ thus serves as a crucial analytical bridge connecting the value function $u$ and the probability density $p$.

To place \eqref{eq:w-equation} within the viscosity-solution framework, we introduce the nonlinear elliptic operator
\begin{equation}\label{eq:def_F_w}
F(t,x,D^2 w,w):=\frac12\,\sigma^2(w)\,w_{xx},
\end{equation}
so that \eqref{eq:w-equation} can be rewritten in the compact form
\begin{equation}\label{eq:w_as_F}
w_t+F(t,x,D^2 w,w)=0.
\end{equation}
The operator $F$ is degenerate elliptic with respect to the Hessian
variable, in the sense that 
\[
F(t,x,M,r)\le F(t,x,N,r)\qquad\text{if } \,\, M\le N,
\]
where $\le$ denotes the usual partial ordering on symmetric matrices.
Since the coefficient $\sigma^2(\cdot)$ is piecewise constant, the operator
$F$ is discontinuous in its scalar argument $w$. Following the standard
approach for fully nonlinear equations with discontinuous nonlinearities
(cf.\ Barles--Souganidis~\cite{G1991Convergence}), we introduce the
upper and lower semicontinuous envelopes required to define viscosity solutions.

For any locally bounded function $Z$ on a closed set $C$, define
\[
Z^{*}(x)=\limsup_{y\to x,\;y\in C}Z(y),
\qquad
Z_{*}(x)=\liminf_{y\to x,\;y\in C}Z(y).
\]
Applying this to $F$ with respect to the scalar variable $w$ yields, for \eqref{eq:def_F_w},
\begin{equation}\label{eq:F_star}
F^{*}(t,x, D^2\psi,\psi)=
\begin{cases}
\frac12\,\overline\sigma^{\,2}\,\partial_{xx} \psi, & \psi>0,\\[4pt]
\frac12\,\max\{\overline\sigma^{\,2} \partial_{xx} \psi,\;\underline\sigma^{\,2} \partial_{xx}\psi\}, & \psi=0,\\[4pt]
\frac12\,\underline\sigma^{\,2}\,\partial_{xx} \psi, & \psi<0,
\end{cases}
\end{equation}
and
\begin{equation}\label{eq:F_sub}
F_{*}(t,x, D^2\psi,\psi)=
\begin{cases}
\frac12\,\overline\sigma^{\,2}\,\partial_{xx} \psi, & \psi>0,\\[4pt]
\frac12\,\min\{\overline\sigma^{\,2}\partial_{xx} \psi,\;\underline\sigma^{\,2}\partial_{xx} \psi\}, & \psi=0,\\[4pt]
\frac12\,\underline\sigma^{\,2}\,\partial_{xx} \psi, & \psi<0.
\end{cases}
\end{equation}
In particular, $F^{*}=F_{*}=F$ whenever $\psi\neq 0$.

\begin{definition}[Viscosity sub-/supersolutions.]
A locally bounded function $w^{*}:[0,1]\times\mathbb{R}\to\mathbb{R}$ is called a viscosity
\emph{subsolution} of \eqref{eq:w_as_F} if for every test function
$\psi\in C^{\infty }((0,1)\times\mathbb{R})$ and every point $(t_0,x_0)$ at which
$w^{*}-\psi$ attains a local maximum, one has
\begin{equation}\label{eq:visc_sub_w}
\partial_{t} \psi(t_0,x_0)+F_{*}\!\big(t_0,x_0, D^2\psi(t_0,x_0),\, w^{*}(t_0,x_0)\,\big)\le 0.
\end{equation}
Similarly, $w_{*}$ is a viscosity \emph{supersolution} if for every $\psi$ and every point
$(t_0,x_0)$ at which $w_{*}-\psi$ attains a local minimum, one has
\begin{equation}\label{eq:visc_super_w}
\partial_{t}\psi(t_0,x_0)+F^{*}\!\big(t_0,x_0, D^2\psi(t_0,x_0),\, w_{*}(t_0,x_0)\big)\ge 0.
\end{equation}
\end{definition}
The function $w$ is said to be a (viscosity) solution of \eqref{eq:w-equation}, if it is both sub- and supersolution of 
\eqref{eq:w-equation}.

In the next subsection~\ref{sec:5.2}, we derive a monotone finite difference scheme for \eqref{eq:w-equation} from the discrete scheme \eqref{dbp} satisfied by $U$. We prove that, provided the viscosity solution of \eqref{eq:w-equation} is unique, the numerical solution produced by this scheme converges to it. Moreover, the convergence of the numerical approximation also provides a constructive argument for the existence of the viscosity solution. This result will serve as a key ingredient in establishing the weak convergence of the responsive probability measures to the solution
of the Fokker--Planck equation \eqref{eq:p-Fokker-plank}.

\subsection{A monotone finite difference scheme for the $w$ equation} \label{sec:5.2}


This subsection addresses the key difficulty in the convergence analysis of the responsive distributions, namely, the dependence of the optimal feedback control on the second spatial derivative $u_{xx}$ of the value function. We show that the discrete second-order difference quotients generated by the backward trinomial tree converge locally uniformly to $u_{xx}$. Such a strong convergence result is highly nontrivial for fully nonlinear problems and does not follow from the convergence of the value function alone. To obtain it, we derive an auxiliary monotone scheme for \eqref{eq:w-equation} from the backward trinomial tree. We then deduce the convergence of the discrete second-order difference quotients from the convergence of the auxiliary variable $w=\sigma^2(u_{xx})u_{xx}$.

\paragraph{Derivation of the scheme.}
We begin by recalling the explicit finite difference formulation of the backward trinomial tree derived in Section~\ref{sec:3}:
\[
\delta_t U_i^{n+1}
+\frac12\,\sigma^2\!\big(\delta_h^2 U_i^{n+1}\big)
\,\delta_h^2 U_i^{n+1}=0 .
\]
Let us introduce the discrete curvature and the associated nonlinear flux:
\[
V_i^n := \delta_h^2 U_i^n,
\qquad
W_i^n := \sigma^2(V_i^n)\,V_i^n .
\]
The quantity $V_i^n$ is the discrete analogue of $u_{xx}$, while $W_i^n$ is the discrete counterpart of the auxiliary variable $w=\sigma^2(u_{xx})u_{xx}$. 
Applying the discrete Laplacian $\delta_h^2$ to the scheme for $U$, and using the commutativity of $\delta_h^2$ and $\delta_t$ on the infinite lattice $\mathbb Z$, we obtain
\[
\delta_t V_i^{n+1}
+\delta_h^2
\!\left(
\frac12\,\sigma^2(V_i^{n+1})V_i^{n+1}
\right)=0 .
\]
Since $\sigma^2(\cdot)$ depends only on the sign,
$\sigma^2(V_i^n)=\sigma^2(W_i^n)$ and therefore
\[
V_i^n=\frac{W_i^n}{\sigma^2(W_i^n)} .
\]
Multiplying the above equation by $\sigma^2(W_i^{n+1})$ yields
\begin{equation}\label{eq:sig(W)}
\sigma^2(W_i^{n+1})
\,\delta_t
\!\left(
\frac{W_i^{n+1}}{\sigma^2(W_i^{n+1})}
\right)
+
\frac12\,\sigma^2(W_i^{n+1})\,\delta_h^2 W_i^{n+1}=0 .
\end{equation}
Noting that $\sigma^2(W_i^{n+1})$ remains constant in time at the fixed node
$(n+1,i)$, the equation \eqref{eq:sig(W)} simplifies to
\begin{equation}\label{eq:scheme_w_derived}
\frac{
W_i^{n+1}
-
\dfrac{\sigma^2(W_i^{n+1})}{\sigma^2(W_i^{n})}
\,W_i^{n}
}{\Delta t}
+
\frac{\sigma^2(W_i^{n+1})}{2}
\frac{
W_{i+1}^{n+1}-2W_i^{n+1}+W_{i-1}^{n+1}
}{h^2}
=0 .
\end{equation}
This yields the following explicit finite difference scheme
for the $w$ equation.

\paragraph{Discrete operator form and convergence analysis.}
To analyze this scheme within the viscosity solution framework, we define the piecewise constant interpolation
\[
\begin{aligned}
v^{h,\Delta t}(t,x)=V_i^n,
\qquad
(t,x)\in((n-1)\Delta t,n\Delta t]\times
\Bigl(x_i-\tfrac{h}{2},x_i+\tfrac{h}{2}\Bigr],\\
w^{h,\Delta t}(t,x)=W_i^n,
\qquad
(t,x)\in((n-1)\Delta t,n\Delta t]\times
\Bigl(x_i-\tfrac{h}{2},x_i+\tfrac{h}{2}\Bigr].
\end{aligned}
\]
Motivated by the probabilistic structure of the trinomial tree, we can reformulate the derived scheme \eqref{eq:scheme_w_derived} into a explicit operator form:
\begin{equation}\label{eq:scheme_w_operator}
L^{h,\Delta t}\!\left(t_n,x_i,\,w^{h,\Delta t}(t_n,x_i),\,w^{h,\Delta t}\right)=0,
\qquad n=0,\dots,N-1,\ i\in\mathbb{Z},
\end{equation}
where the discrete operator \(L^{h,\Delta t}\) is defined as
\begin{equation}\label{eq:scheme_w_def}
\begin{aligned}
L^{h,\Delta t}(t,x,\psi(t,x),\psi)
:=
&\Biggl[
\psi(t+\Delta t,x)
-
\frac{\sigma^{2}(\psi(t+\Delta t,x))}{\sigma^{2}(\psi(t,x))}\,\psi(t,x)
\Biggr]\Big/\Delta t \\
&\,\,
+\frac{\sigma^{2}(\psi(t+\Delta t,x))}{2}\,
\frac{\psi(t+\Delta t,x+h)-2\psi(t+\Delta t,x)+\psi(t+\Delta t,x-h)}{h^{2}} .
\end{aligned}
\end{equation}
Finally, matching the terminal condition $u(1,x)=\phi(x)$ implies $u_{xx}(1,x)=\phi''(x)$. Thus, the terminal condition for the $w$-scheme is consistently initialized as:
\begin{equation}\label{eq:scheme_w_terminal}
W_i^{N} = w(1,x_i)=\sigma^{2}(\phi''(x_i))\,\phi''(x_i).
\end{equation}

\begin{assumption}[Comparison principle.]
\label{ass:comparison_w}
The fully nonlinear equation \eqref{eq:w-equation} satisfies a strong
comparison (uniqueness) principle in the class of bounded viscosity
solutions. More precisely, if
$w_1$ is an upper semicontinuous subsolution and $w_2$ is a lower
semicontinuous supersolution of \eqref{eq:w-equation} on
$[0,1]\times\overline\Omega$, then
\[
w_1\le w_2 \qquad \text{on } [0,1]\times\overline\Omega .
\]
\end{assumption}

\begin{lemma}[{Monotonicity.}]
\label{lem:mono_w}
Under stability condition \eqref{eq:CFL_w},
for each \(n\in\{0,\dots,N-1\}\) and \(i\in \mathbb{Z}\), let \(W_i^n\) denote the numerical approximation at grid point \((t_n, x_i)\) defined by the scheme
\begin{equation}\label{eq:scheme_mono}
W_i^n = L^{h,\Delta t}(t_n,x_i,w^{h,\Delta t}(t_n,x_i),w^{h,\Delta t}).
\end{equation}
If \(\widetilde{W}_i^{n+1} \leq W_i^{n+1}\) for all \(i\), then \(\widetilde{W}_i^n \leq W_i^n\) for all \(i\).
\end{lemma}
The proof is given in Appendix~\ref{app:proof_mono_w}.

\begin{lemma}[Stability.]
Under the stability condition \eqref{eq:CFL_w}, the scheme is $\ell^\infty$-stable:
\begin{equation}\label{eq:stab_w}
\sup_{n,i}|W_i^n|\ \le\ \sup_i |W_i^{N}|.
\end{equation}
\end{lemma}

\begin{definition}
The numerical scheme \eqref{eq:scheme_w_def} is said to be \emph{consistent} with the continuous equation \eqref{eq:w-equation} if and only if, for any $(t_0, x_0) \in (0,T) \times \Omega$ and any test function $\psi \in C_c^{\infty}([0,T) \times \Omega )$, the following holds:
\small
\begin{equation}\label{eq:scheme_consist}
\begin{aligned}
\limsup_{\substack{h,\Delta t \to 0,\, \xi \to 0,\\ (\tau,y) \to (t_0,x_0)}}
\mathcal{L}^{h,\Delta t}\big(\tau,y,\psi^{h,\Delta t}(\tau,y) + \xi, \psi + \xi\big)
&\ge
\limsup_{\substack{(\tau,y) \to (t_0,x_0) \\ (\tau,y)\in[0,T)\times\Omega}}
\Big(
\partial_t \psi(\tau,y)
+ \tfrac{1}{2}\sigma^2(\psi)\,\partial_{xx}\psi(\tau,y)
\Big), \\[4pt]
\liminf_{\substack{h,\Delta t \to 0,\, \xi \to 0,\\ (\tau,y) \to (t_0,x_0)}}
\mathcal{L}^{h,\Delta t}\big(\tau,y,\psi^{h,\Delta t}(\tau,y) + \xi, \psi + \xi\big)
&\le
\liminf_{\substack{(\tau,y) \to (t_0,x_0) \\ (\tau,y)\in[0,T)\times\Omega}}
\Big(
\partial_t \psi(\tau,y)
+ \tfrac{1}{2}\sigma^2(\psi)\,\partial_{xx}\psi(\tau,y)
\Big).
\end{aligned}
\end{equation}
\end{definition}

\begin{lemma}[Consistency.]
Let $\psi\in C_c^{\infty}([0,1) \times \mathbb{R})$ be a smooth test function. Then the scheme
\eqref{eq:scheme_w_def} is consistent with \eqref{eq:w-equation} in the viscosity sense:
\begin{equation}\label{eq:cons_w}
\lim_{h,\Delta t\to0}\,
L^{h,\Delta t}(t,x,\psi(t,x),\psi)
=
\partial_t\psi(t,x)+\frac12\sigma^2(\psi(t,x))\,\partial_{xx}\psi(t,x),
\end{equation}
with the upper/lower semicontinuous envelopes of the nonlinear term at $\psi(t,x)=0$.
\end{lemma}
The detailed verification is given in Appendix~\ref{app:proof_cons_w}. 

\begin{theorem}[Convergence of the auxiliary monotone scheme]
\label{thm:conv_w}
Under the Assumptions \ref{ass:CFL_w} and \ref{ass:comparison_w}, let $W_i^n$ be the solution of the discrete scheme
\eqref{eq:scheme_w_operator}--\eqref{eq:scheme_w_terminal}, and let $w^{h,\Delta t}$ be its
piecewise constant interpolation. Then, as $h,\Delta t\to0$,
\[
w^{h,\Delta t}\ \longrightarrow\ w
\quad\text{locally uniformly on }(0,1)\times\mathbb{R},
\]
where $w$ is the unique continuous viscosity solution of \eqref{eq:w-equation}.
\end{theorem}

\begin{proofsketch}
By \eqref{eq:scheme_mono}, \eqref{eq:stab_w}, and \eqref{eq:cons_w}, the scheme is monotone, $\ell^\infty$-stable,
and consistent.  Therefore, the convergence extends from the general framework of
Barles--Souganidis {\cite{G1991Convergence}}, adapted to the possibly discontinuous nonlinearity
$\sigma^2(\cdot)$ via upper/lower semicontinuous envelopes.  The full proof (including the limsup/liminf
argument and the identification of the limiting sub-/supersolutions) is given in
Appendix~\ref{app:proof_conv_w}.
\end{proofsketch}

\begin{theorem} \label{thm:conv_v}
Since $v^{h,\Delta t} := w^{h,\Delta t} / \sigma^2(w^{h,\Delta t})$, the discrete second-order difference $v^{h,\Delta t}$ converges to the exact second derivative $u_{xx}$ locally uniformly. Specifically, for every compact set $K \subset (0,1)\times\mathbb{R}$, we have:
\[
\sup_{(t,x)\in K} \bigl| v^{h,\Delta t} - u_{xx} \bigr| \longrightarrow 0, \quad \text{as } h, \Delta t \to 0.
\]
The full proof is given in
Appendix~\ref{app:proof_conv_v}.
\end{theorem}

The strong convergence of the discrete second-order difference obtained in Theorem \ref{thm:conv_v}
implies strong convergence of the corresponding discrete diffusion coefficients.

\begin{lemma}\label{lem:sig_V_to_sig_uxx}
$\sigma^2$ has only finitely many discontinuities and is bounded,
the convergence $V_i^n\to \partial_{xx} u(t_n,x_i)$ implies
\begin{equation}\label{eq:sigma_strong}
\sigma^2(V_i^n)\ \longrightarrow\ \sigma^2(\partial_{xx} u(t_n,x_i))
\quad\text{strongly in }L^p([0,T]\times \overline{\Omega})\ \text{for any }1\le p<\infty,
\end{equation}
by Vitali's theorem \cite{folland1999real} (details are given in Appendix~\ref{app:vitali_sigma}).
\end{lemma}

The establishment of this local uniform convergence for $V_i^n \to \partial_{xx} u$ constitutes a fundamental theoretical milestone. By guaranteeing that the discrete optimal control $\sigma^2(V_i^n)$ converges strongly to the exact continuous optimal control $\sigma^2(\partial_{xx} u)$, this result strictly bridges the gap between the backward value propagation and the forward density evolution. We now leverage this vital property to rigorously pass to the limit for the discrete probability measures generated by the forward trinomial tree.

\subsection{Weak$^*$ Convergence of the Responsive Distribution}
\label{subsec:conv_distribution}


We now establish the main convergence result for the forward trinomial tree.  
The discrete probability masses propagated by the forward recursion define a sequence of discrete measures that serve as approximate laws for the continuous optimally controlled diffusion $X_t^*$.
Building on the strong convergence
of the discrete optimal feedback control obtained in the previous subsection, we show that these
discrete laws converge, in the weak$^\ast$ sense, to the responsive distribution. Equivalently,
the limit is a weak solution of the Fokker--Planck equation \eqref{eq:p-Fokker-plank}.

\paragraph{From trinomial transitions to a finite difference scheme.}
Let $x_i=ih$ and $t_n=n\Delta t$. Denote by $p_i^n$ the discrete density at $(t_n,x_i)$ and
set $\sigma_i^n := \sigma\!\big(V_i^n\big)$, where $V_i^n$ is the discrete approximation of
$v=\partial_{xx} u$ obtained in Theorem~\ref{thm:conv_v}. Then the forward recursion of the trinomial
tree \eqref{eq:sampling2_p} can be written equivalently as the explicit scheme
\begin{equation}\label{eq:p_fdm}
-\,\frac{p_i^{\,n+1}-p_i^{\,n}}{\Delta t}
\;+\;
\frac{1}{2}\,
\frac{ \big(\sigma_{i+1}^{\,n}\big)^{2} p_{i+1}^{\,n}
-2\big(\sigma_{i}^{\,n}\big)^{2} p_{i}^{\,n}
      +\big(\sigma_{i-1}^{\,n}\big)^{2} p_{i-1}^{\,n} }{h^{2}}
\;=\;0,
\end{equation}
which is a finite difference approximation of
$\partial_{t} p+\partial_{xx}\!\big(\frac12 \sigma^2(\partial_{xx} u)\,p\big)=0$ in \eqref{eq:p-Fokker-plank}.
We prescribe the discrete initial condition corresponding to the Dirac measure $p(0,\cdot)=\delta_0$ by setting
$p_0^0=1$ and $p_i^0=0$ for $i\neq 0$. Consequently, $p_i^n=0$ whenever $|i|>n$, so that the nonzero entries remain confined to the inverted triangular region $i=-n,\dots,n$.

The discrete scheme preserves the fundamental properties of a probability measure throughout the evolution, as stated in Proposition~\ref{lem:probability_preserving}.

\begin{proposition}[Conservation of total probability]\label{lem:probability_preserving}
Let $\{p_i^n\}_{i\in\mathbb Z}$ be generated by the trinomial forward recursion
\eqref{eq:sampling2_p}. Under the stability condition \eqref{ass:CFL_w}
for every $n\ge 0$ the sequence $\{p_i^n\}$ defines a discrete probability measure:
$p_i^n\ge 0$ for all $i$ and
\begin{equation}\label{eq:probability}
\sum_{i=-\infty}^{\infty} p_i^n\,h \;=\; 1.
\end{equation}
\end{proposition}
The full proof is given in
Appendix~\ref{app:proof_probability_preserving}.

\begin{lemma} \label{lemma:Radon}
Assume that the solution to equation \eqref{eq:weak_form} exists and is unique. Let $\{p_n\}$ be a nonnegative bounded sequence in $L^{1}(B)$. 
Then, $\{p_n\}$ is weakly* compact in $M(B)$, i.e., there exists a subsequence $u_{n_k}$ and $u\in M(B)$ such that
\begin{equation}\label{eq:2.16}
\lim_{k\to\infty}\int_{B} p_{n_k}\,\phi\,dx
  = \langle p,\phi\rangle_{M(B),\,C_c^0(B)}
  \quad \text{for any } \phi\in C_c^0(B),
\end{equation}
where $M(B)\subset \big(C_c^0(B)\big)'\,$ is the set of nonnegative linear functionals on $C_c^0(B)$, i.e., Radon measures.
\end{lemma}


\begin{theorem}[Weak$^\ast$ convergence to the responsive distribution]
\label{thm:weakstar_FP}
Assume that the stability condition \eqref{ass:CFL_w} holds.
Let $\{p_i^n\}_{i\in\mathbb Z,\,0\le n\le N}$ be generated by the forward trinomial recursion
\eqref{eq:sampling2_p}, equivalently by the finite difference scheme \eqref{eq:p_fdm}. We define the piecewise constant interpolation $p^{h,\Delta t}$ on $(0,T)\times\mathbb R$ by
\[
p^{h,\Delta t}(t,x)=p_i^n,
\qquad
(t,x)\in((n-1)\Delta t,n\Delta t]\times
\Bigl(x_i-\tfrac{h}{2},x_i+\tfrac{h}{2}\Bigr].
\]
Then there exist a subsequence (not relabeled) and a nonnegative Radon measure
$p\in M([0,T)\times \mathbb R)$ such that
\begin{equation}
\begin{aligned}
&p^{h,\Delta t} \to p(t,x) \in M([0,T)\times \mathbb R) \quad weakly^* \; in \; M([0,T)\times \mathbb R),\\
&p^{h,T} \to p(T,x) \in M(\mathbb R) \quad weakly^* \; in \; M(\mathbb R),
\end{aligned}
\end{equation}
Moreover, the limit $p$ is a weak solution of the Fokker--Planck equation \eqref{eq:p-Fokker-plank}
on $(0,T)\times\mathbb R$ with initial datum $p(0,\cdot)=\delta(x-0)$, i.e., for every
${\varphi}\in C_c^{\infty}([0,T]\times \mathbb R )$,
\begin{equation}\label{eq:weak_FP}
\begin{aligned}
&\int_0^T \int_{\Omega}
   p(t,x)\Big(
      \partial_t \varphi(t,x)
      + \frac{\sigma^2(u_{xx}(t,x))}{2}
        \partial_{xx}\varphi(t,x)
   \Big)\,dx\,dt  \\
& \qquad\qquad\;\;\;\; + \langle \delta(x-x_0), \varphi(0,x) \rangle
   - \int_{\Omega} p(T,x)\varphi(T,x)\,dx =0.
\end{aligned}
\end{equation}

\end{theorem}


\begin{proofsketch}
By Proposition~\ref{lem:probability_preserving}, the family $\{p^{h,\Delta t}\}$ is nonnegative and uniformly bounded in total variation. Hence, up to a subsequence, there exist
\[
p\in M([0,T)\times\mathbb R),
\qquad
p(T,\cdot)\in M(\mathbb R),
\]
such that
\[
p^{h,\Delta t}\overset{\ast}{\rightharpoonup} p
\quad\text{in }M([0,T)\times\mathbb R),
\qquad
p^{h,T}\overset{\ast}{\rightharpoonup} p(T,\cdot)
\quad\text{in }M(\mathbb R).
\]
To identify the limit, we multiply the conservative scheme \eqref{eq:p_fdm} by a smooth test function and perform discrete summation-by-parts in both time and space. This yields a discrete weak formulation of the forward scheme. Passing to the limit in this identity uses two ingredients: the weak$^\ast$ convergence of $p^{h,\Delta t}$ and the strong $L^p$ convergence
\[
\sigma^2(V^{h,\Delta t}) \to \sigma^2(u_{xx}),
\]
established in Lemma~\ref{lem:sig_V_to_sig_uxx}. Consequently, every weak$^\ast$ limit point satisfies the weak formulation \eqref{eq:weak_FP} of the Fokker--Planck equation \eqref{eq:p-Fokker-plank}.
Finally, by uniqueness of weak solutions to \eqref{eq:p-Fokker-plank}, the whole sequence converges, and the limit coincides with the responsive distribution of the optimally controlled diffusion.
The full proof is given in Appendix~\ref{app:proof_thm_weakstar_FP}.
\end{proofsketch}

\section{Numerical Experiments}\label{sec:6}
We present numerical experiments to illustrate the theoretical results developed in the previous sections and to assess the performance of the coupled trinomial tree schemes. The experiments have three objectives. First, we compute responsive distributions under several measurement functions to highlight the intrinsically measurement-dependent nature of G-normal distributions. Second, we examine the convergence of the discrete second-order difference quotients and the associated nonlinear flux, which underpins the convergence of the discrete optimal control. Third, we investigate the convergence of the responsive distributions generated by the forward trinomial tree through systematic grid refinement.

Throughout all numerical experiments, we employ a consistent parameter configuration to ensure reproducibility and enable meaningful comparisons across different scenarios. The G-normal parameters are set to $\underline{\sigma}^2 = 0.04$ and $\overline{\sigma}^2 = 1.0$, representing substantial volatility uncertainty with a variance range spanning nearly two orders of magnitude. This parameter choice reflects realistic financial market conditions where volatility estimates exhibit significant uncertainty. The time horizon is fixed at $T = 1$, and the discretization parameters are systematically selected to satisfy the stability constraint $\frac{\overline{\sigma}^2 \Delta t}{h^2} \leq 1$ established in Assumption~\ref{ass:CFL_w}, ensuring numerical stability throughout all computations.

\subsection{Sampling for G-normal Distributions}
The probability distribution of the G-normal is generally uncertain. However, through equation \eqref{primary_value_f_1}, we know that once a specific observation $\phi$ is given, the responsive distribution of the G-normal can be determined. This is like a quantum state, which is uncertain and unknown before
measurement. Given different measurements of $\phi$, we can obtain different information. Here, we implement $\mathbb{E}[\phi(X)]$ using a trinomial tree and obtain the probability density function of G-normal {under a test function}. In this section, we take $\phi = x^2, -x^2, \sin(3x), x^3$ as examples and sample the probability distribution of the G-normal under these test functions. The results are presented in Figures \ref{fig_x^2_-x^2} -- \ref{fig_x^3}.

\begin{figure*}[ht]
	\centering
	\begin{minipage}[t]{0.5\linewidth}
		\centering
		\includegraphics[width=\textwidth]{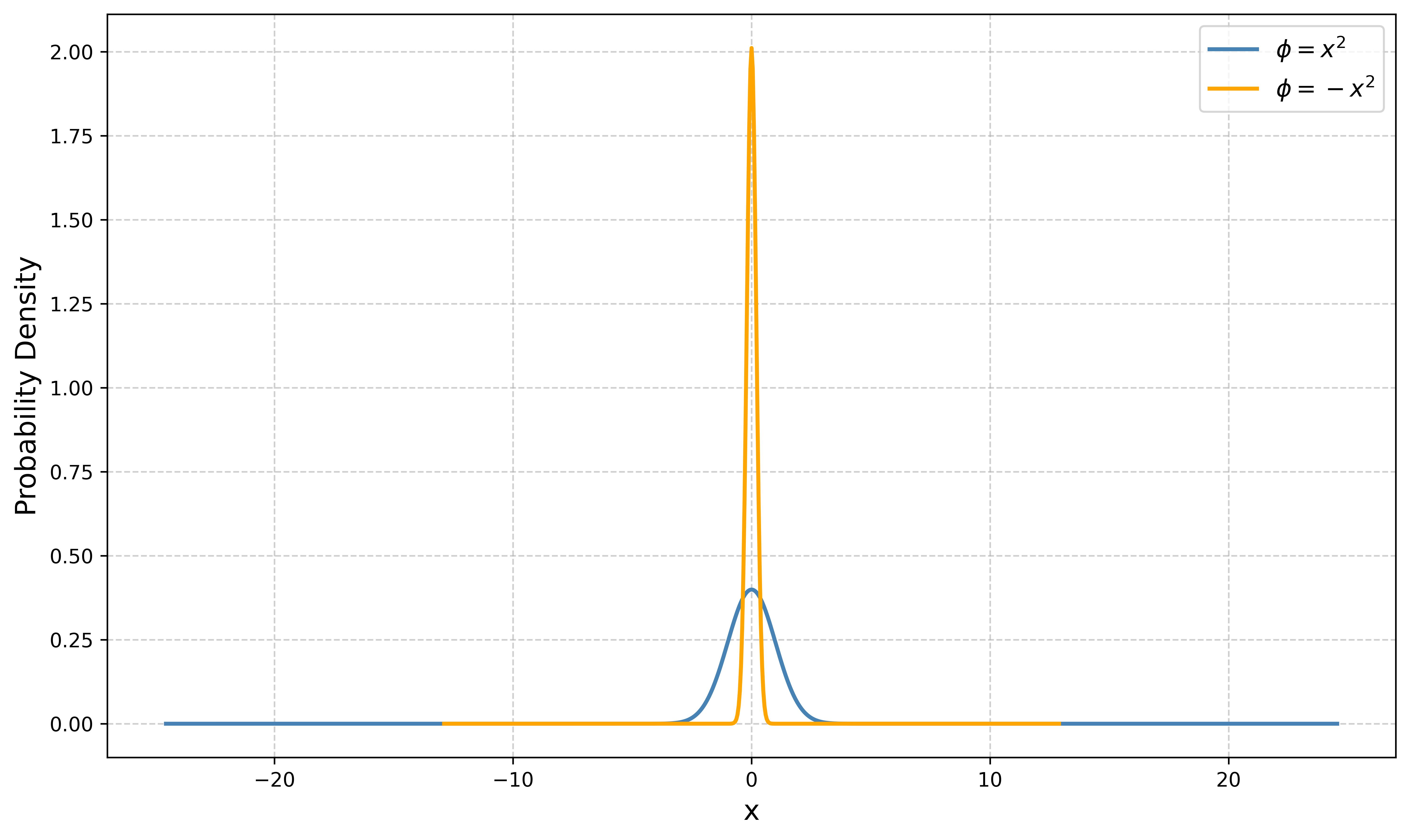}
		\label{fig:x^2}
	\end{minipage}%
	\hfill
	\begin{minipage}[t]{0.5\linewidth}
		\centering
		\includegraphics[width=\textwidth]{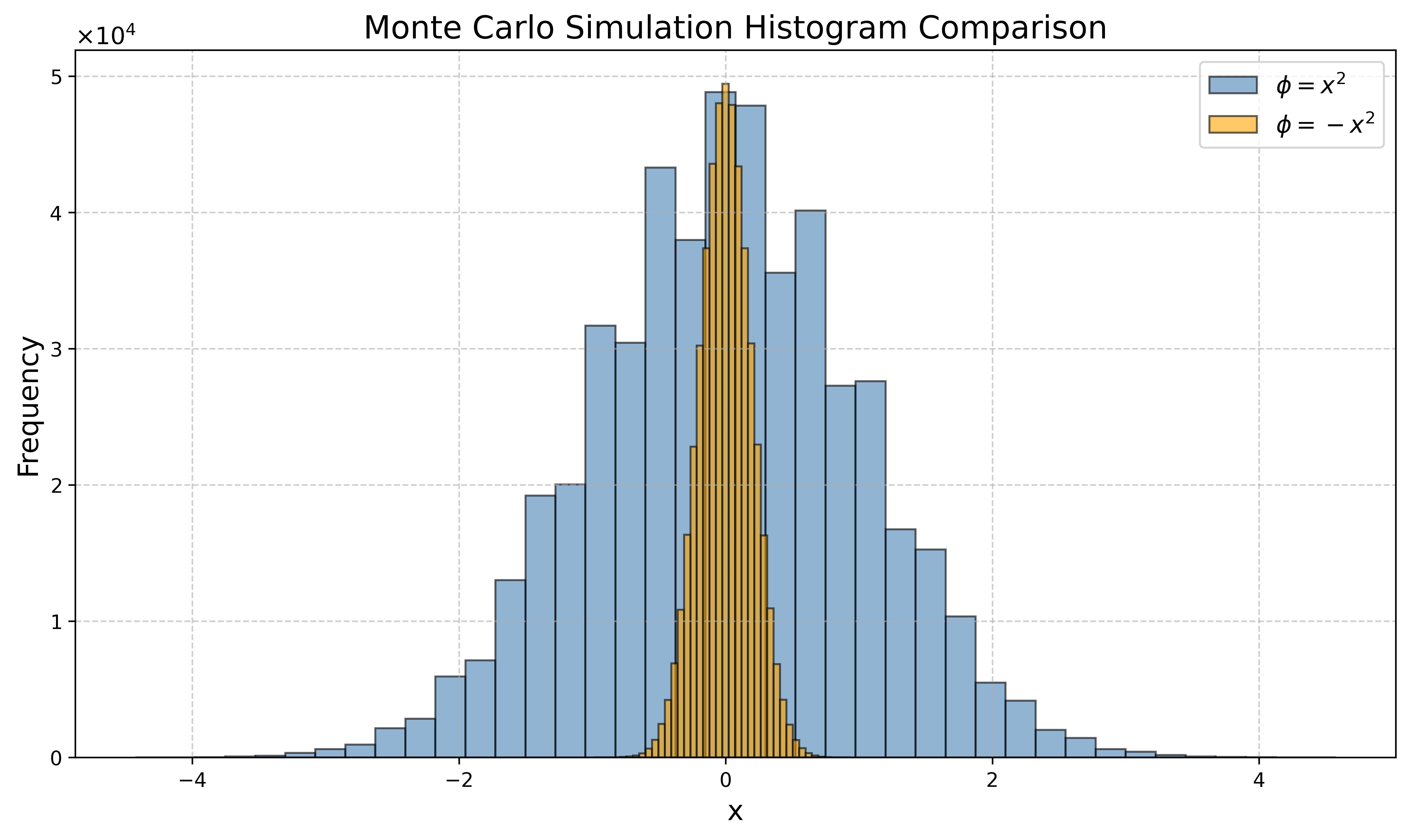}
		\label{fig:-x^2}
	\end{minipage}
	\vspace{-0.7cm}
	\centering
	\caption{The left figure compares the probability density functions of G-normal distributions for measurements using the convex function $x^2$ and the concave function $-x^2$. In this case, the result of convex measurement behaves like a normal distribution $N\left(0, \bar{\sigma}^2\right)$, while the result of concave measurement  behaves like a normal distribution $N\left(0, \underline{\sigma}^2\right)$. The right figure presents the corresponding histogram comparison result when the sample size is 500,000.}
	\label{fig_x^2_-x^2}
\end{figure*}

\begin{figure*}[ht]
	\centering
	\begin{minipage}[t]{0.5\linewidth}
		\centering
		\includegraphics[width=\textwidth]{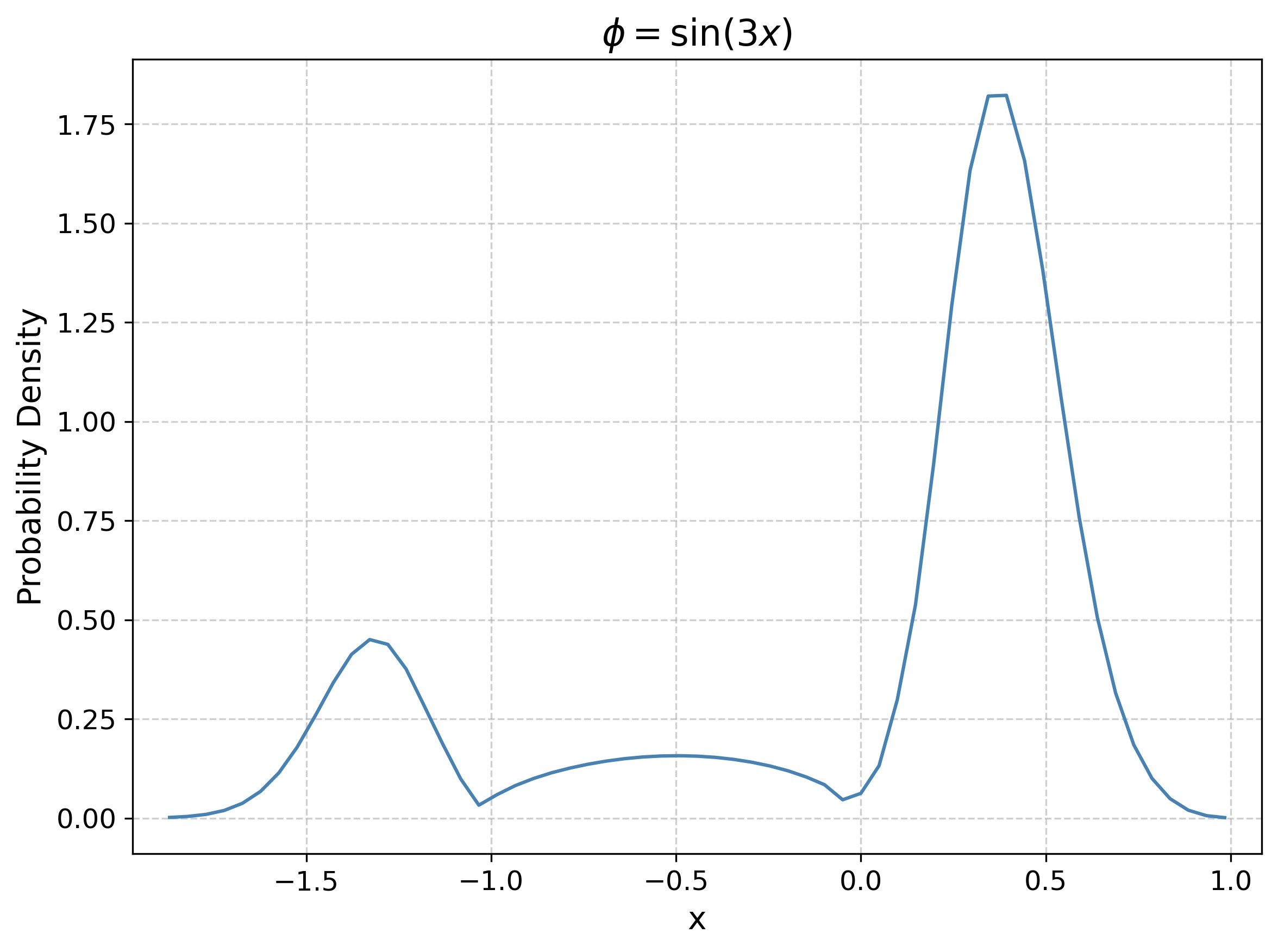}
		\label{fig:x^2}
	\end{minipage}%
	\hfill
	\begin{minipage}[t]{0.5\linewidth}
		\centering
		\includegraphics[width=\textwidth]{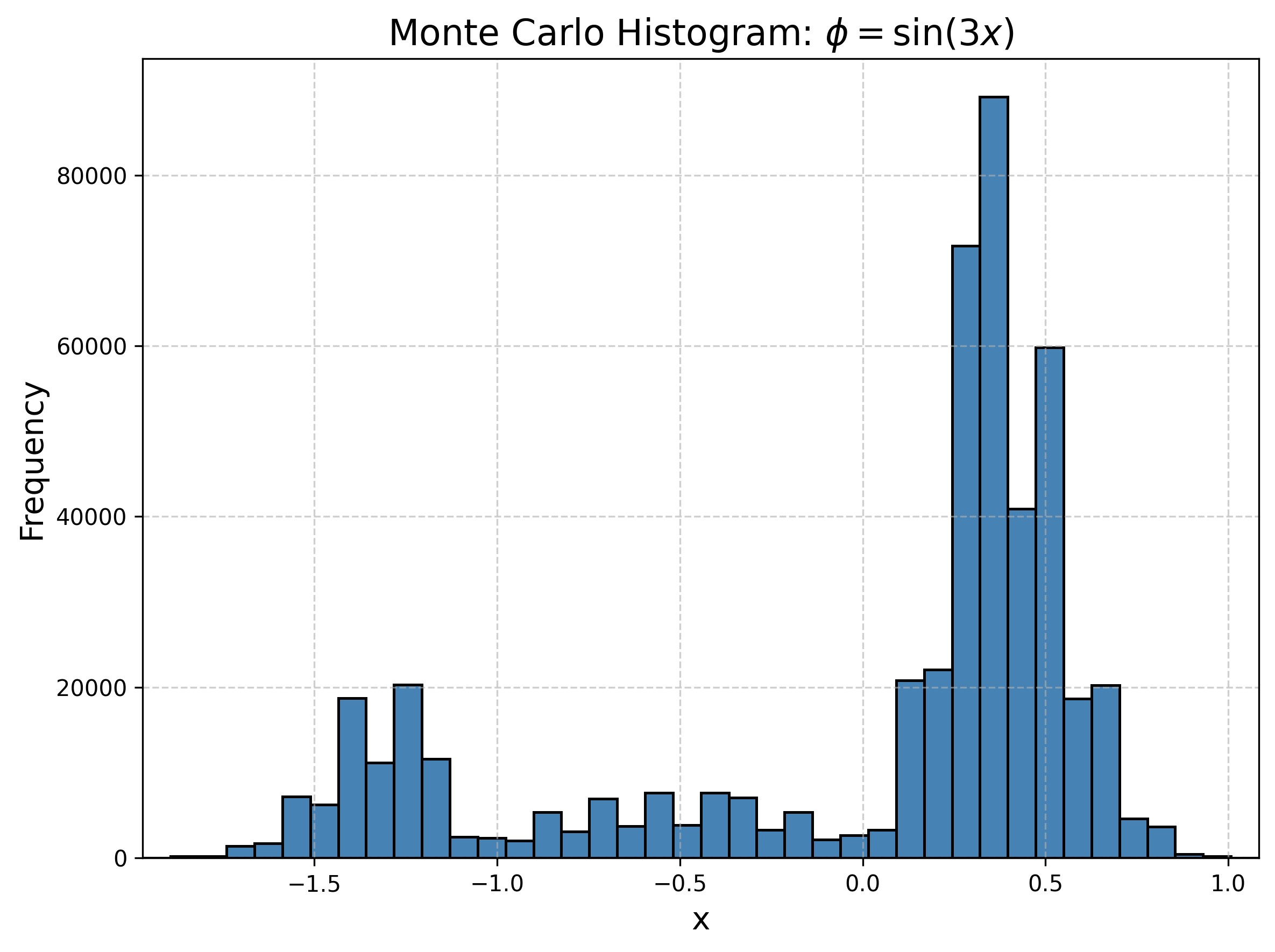}
		\label{fig:-x^2}
	\end{minipage}
	\vspace{-0.7cm}
	\centering
	\caption{The left figure illustrates the probability density function of G-normal distributions when the given measurement is $\sin(3x)$, which exhibits varying convexity and concavity. The right figure presents the corresponding histogram when the number of samples is 500,000.}
	\label{fig_sin(3x)}
\end{figure*}

\begin{figure}
	\centering
	\begin{minipage}[t]{0.5\linewidth}
		\centering
		\includegraphics[width=\textwidth]{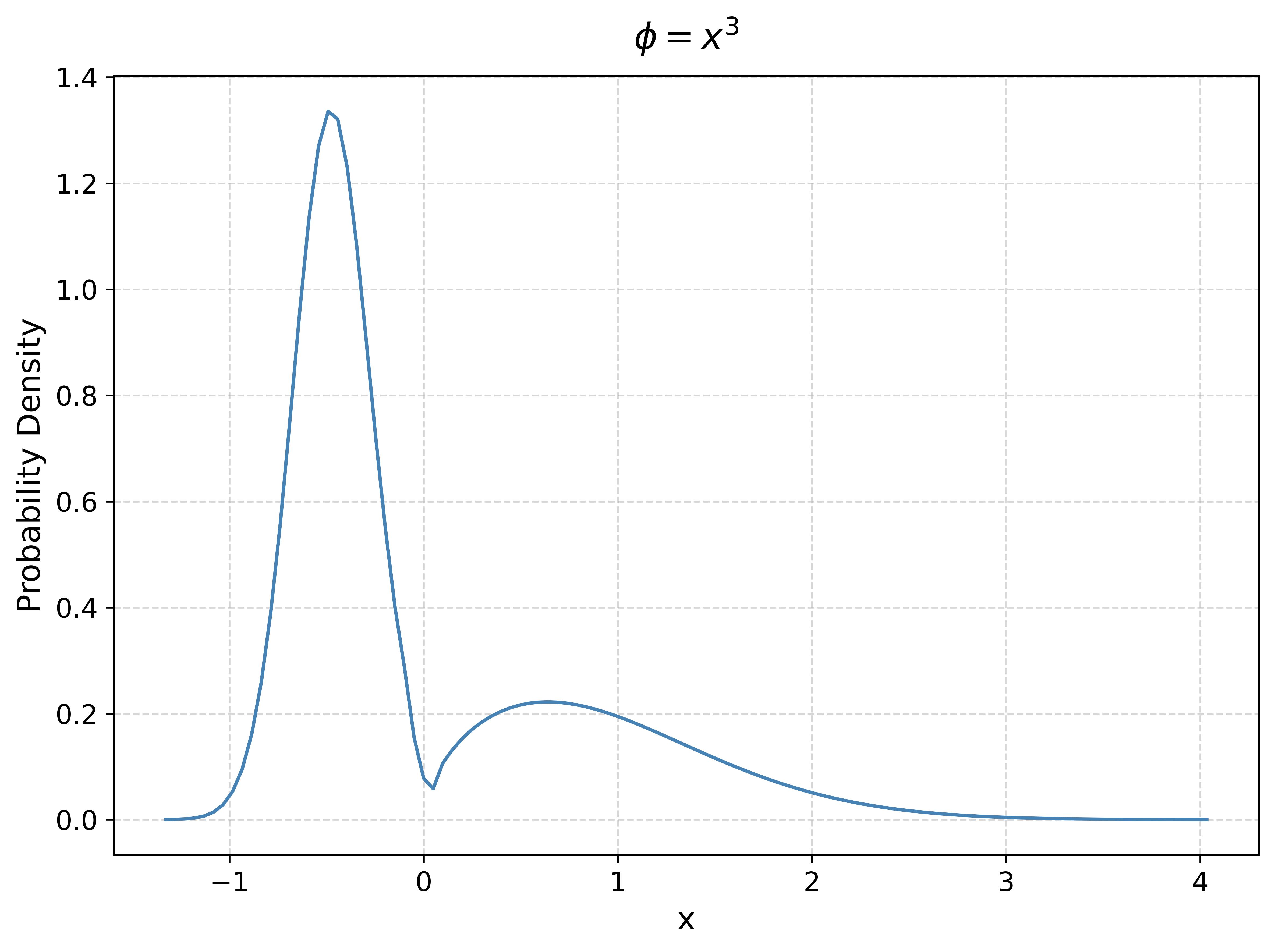}
		\label{fig:x^2}
	\end{minipage}%
	\hfill
	\begin{minipage}[t]{0.5\linewidth}
		\centering
		\includegraphics[width=\textwidth]{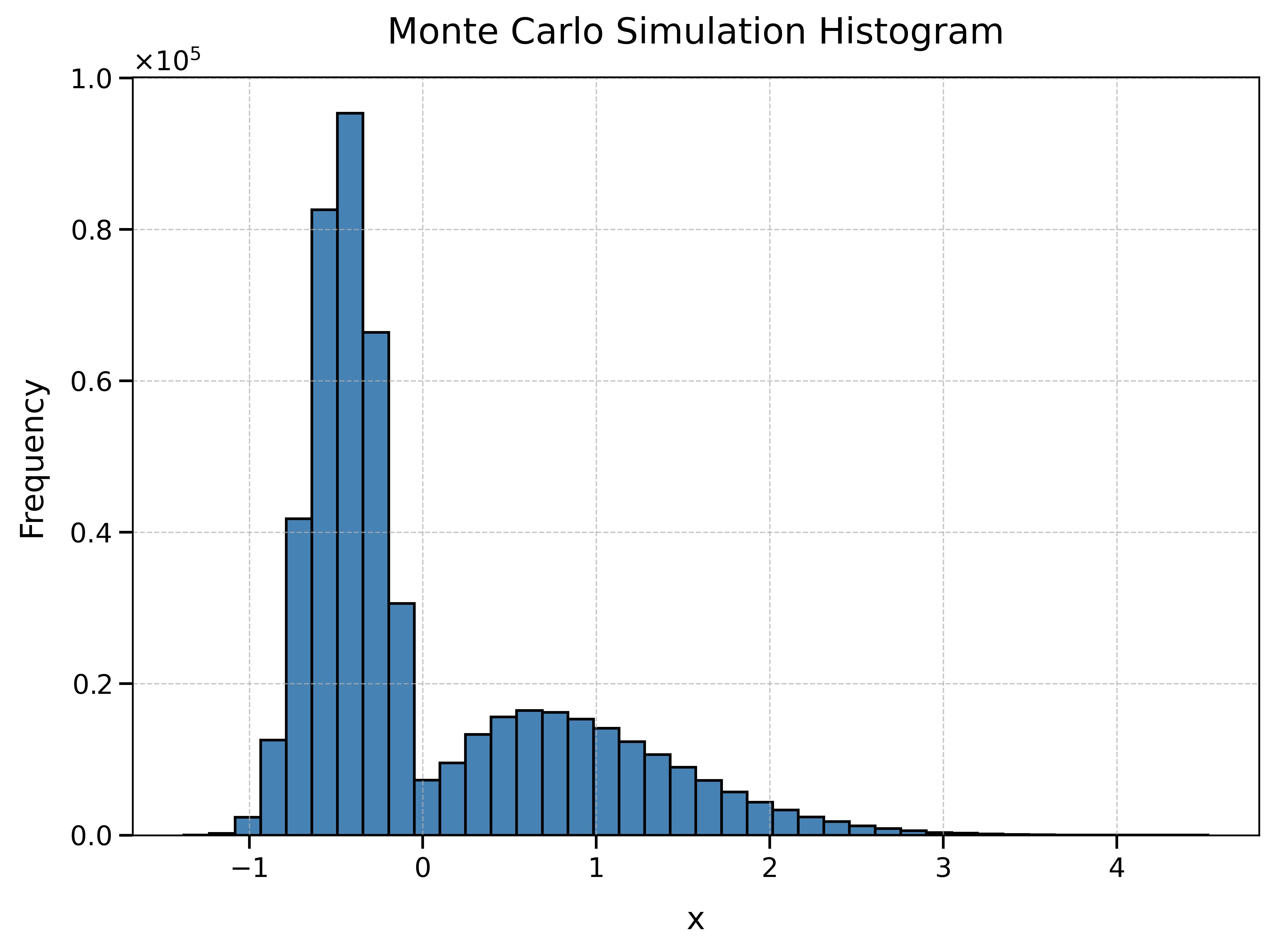}
		\label{fig:-x^2}
	\end{minipage}
	\vspace{-0.7cm}
	\centering
	\caption{The left figure illustrates the probability density function of G-normal distributions when the given measurement is $x^3$, which exhibits varying convexity and concavity. The right figure presents the corresponding histogram when the number of samples is 500,000.}
	\label{fig_x^3}
\end{figure}

\subsection{Convergence of the discrete second-order difference quotient and nonlinear flux}

In this subsection, we provide numerical evidence for the key convergence result established in Section~\ref{sec:5.2}, namely, the convergence of the discrete second-order difference quotients and of the associated nonlinear flux generated by the backward trinomial tree. These quantities play a central role in the analysis, since they determine the discrete optimal feedback control and thereby govern the forward propagation of the responsive distribution.

Recall that, in the notation of Section~\ref{sec:5.2}, the discrete second-order difference quotient is defined by
\[
V_i^n := \delta_h^2 U_i^n
      = \frac{U_{i+1}^n - 2U_i^n + U_{i-1}^n}{h^2},
\]
and the associated nonlinear flux is given by
\[
W_i^n := \sigma^2(V_i^n)\,V_i^n,
\]
where
\[
\sigma^2(z)=
\begin{cases}
\overline{\sigma}^{\,2}, & z\ge 0,\\[2mm]
\underline{\sigma}^{\,2}, & z<0.
\end{cases}
\]
Thus, $V_i^n$ and $W_i^n$ are the discrete counterparts of $u_{xx}$ and
$w=\sigma^2(u_{xx})u_{xx}$, respectively.

For this experiment, we take the terminal condition $\phi(x)=\sin(3x)$ and set $T=1$. The mesh parameters are chosen according to the parabolic scaling
\[
h=\overline{\sigma}\sqrt{\Delta t}\cdot \mathrm{ratio},
\qquad
\mathrm{ratio}=1.1,
\]
which is consistent with the stability regime used throughout the paper. The switching tolerance is fixed at $10^{-6}$.

To assess convergence, we evaluate the numerical errors at the intermediate time $t=0.5$. A reference solution is computed on a sufficiently fine grid with $N_{\mathrm{ref}}=3200$ time steps. For the coarser meshes $N=100,200,400,800$, we measure the $L^\infty$ errors
\[
\|V^{(N)}-V^{(\mathrm{ref})}\|_{L^\infty},
\qquad
\|W^{(N)}-W^{(\mathrm{ref})}\|_{L^\infty},
\]
where the reference solution is interpolated onto the corresponding coarse grid. The results are reported in Table~\ref{tab:curvature_convergence}.

The data show a clear decrease in both errors under mesh refinement, which is consistent with the convergence theory developed in Section~4.2. In particular, the experiment supports the numerical convergence of the discrete second-order difference quotient $V_i^n$ to $u_{xx}$ and of the nonlinear flux $W_i^n$ to $w$. This is especially significant in view of the theoretical role of these quantities: the convergence of $V_i^n$ yields the convergence of the discrete optimal control, while the convergence of $W_i^n$ provides numerical confirmation of the auxiliary monotone scheme introduced for the $w$-equation. Together, these observations furnish numerical support for the analytical mechanism underlying the convergence of the responsive distributions. 

\begin{table}[h]
\centering
\begin{tabular}{c c c c c}
\hline
$N$ & $\|\delta_h^2 U - \delta_h^2 U^{\mathrm{ref}}\|_\infty$
& order &
$\|W - W^{\mathrm{ref}}\|_\infty$
& order \\
\hline
100 & $7.02\times 10^{-1}$ & -- 
    & $5.62\times 10^{-2}$ & -- \\
200 & $6.52\times 10^{-1}$ & 0.21 
    & $4.35\times 10^{-2}$ & 0.74 \\
400 & $3.05\times 10^{-1}$ & 2.20 
    & $2.62\times 10^{-2}$ & 1.47 \\
800 & $1.93\times 10^{-1}$ & 1.31 
    & $1.76\times 10^{-2}$ & 1.15 \\
\hline
\end{tabular}
\caption{$L^\infty$ convergence of the discrete curvature
and nonlinear flux at $t=0.5$.}
\label{tab:curvature_convergence}
\end{table}

\subsection{Numerical Verification of the Convergence for the Responsive Distribution}\label{sec:5.3}

To complement the theoretical convergence results established in Section~\ref{sec:5}, we investigate the numerical convergence behavior of the responsive distributions generated by the forward trinomial tree scheme. The purpose of this experiment is not only to assess numerical accuracy, but also to illustrate how the discrete response–induced distributions approach the weak solution of the associated Fokker–Planck equation as the discretization is refined.


We choose the test function $\phi(x) = \sin(3x)$, which is neither convex nor concave. As discussed in Section~3, such measurements activate the bang–bang structure of the optimal volatility control and therefore induce spatially heterogeneous diffusion. This choice allows us to examine the convergence of the responsive distribution in a setting where the effective diffusion coefficient switches across space.

Let $f_h(x)$ denote the probability density obtained from the forward trinomial tree with spatial step size $h$, and let $f_{\text{ref}}(x)$ be a reference density computed on a sufficiently fine grid ($N = 3200$). We measure the discrepancy between these densities using the $L^2$-norm,
\begin{equation}
    E_{L^2}(h) := \|f_h - f_{\text{ref}}\|_{L^2} = \sqrt{\int_{-\infty}^{\infty} |f_h(x) - f_{\text{ref}}(x)|^2 dx}, \label{eq:L2_error}
\end{equation}
where the integration is restricted to a bounded region capturing the effective support of the distributions. The empirical convergence rate between successive refinements is computed as:
\begin{equation}
    \text{Rate} = \frac{\log(E_{L^2}(h_{\text{coarse}})/E_{L^2}(h_{\text{fine}}))}{\log(h_{\text{coarse}}/h_{\text{fine}})}, \label{eq:conv_rate}
\end{equation}
providing a quantitative measure of the convergence order that can be compared with theoretical predictions for explicit finite difference schemes.

Table~\ref{tab:convergence} reports the convergence behavior under successive grid refinements, with the number of time steps increasing from \(N=100\) to \(N=800\). The \(L^2\)-error decreases monotonically as the grid is refined, indicating the stability and convergence of the forward trinomial recursion for the responsive distribution. This behavior is consistent with the convergence theory developed in Sections~\ref{sec:5}.
Moreover, the observed convergence rates approach second order as the discretization becomes finer.


\begin{table}[htbp]
\centering
\caption{Convergence analysis of probability density functions under the measurement $\phi(x) = \sin(3x)$. The systematic reduction in $L^2$ error demonstrates second-order convergence behavior, validating the theoretical equivalence with explicit finite difference discretization.}
\label{tab:convergence}
\begin{tabular}{cccc}
\toprule
$N$ & $h$ & $E_{L^2}$ & Convergence Rate \\
\midrule
100 & 0.110000 & 5.949$\times 10^{-2}$ & --- \\
200 & 0.077782 & 3.448$\times 10^{-2}$ & 1.574 \\
400 & 0.055000 & 1.739$\times 10^{-2}$ & 1.975 \\
800 & 0.038891 & 8.566$\times 10^{-3}$ & 2.042 \\
\bottomrule
\end{tabular}
\end{table}


Figure~\ref{fig:convergence_plot} provides visual confirmation of the convergence behavior, illustrating the systematic approach of computed probability densities toward the reference solution. The progressive refinement demonstrates that even moderate grid resolutions capture the essential features of the G-normal distribution under complex measurements, while finer discretizations achieve an increasingly precise approximation of subtle distributional characteristics.

\begin{figure}[htbp]
    \centering
    \includegraphics[width=0.75\textwidth]{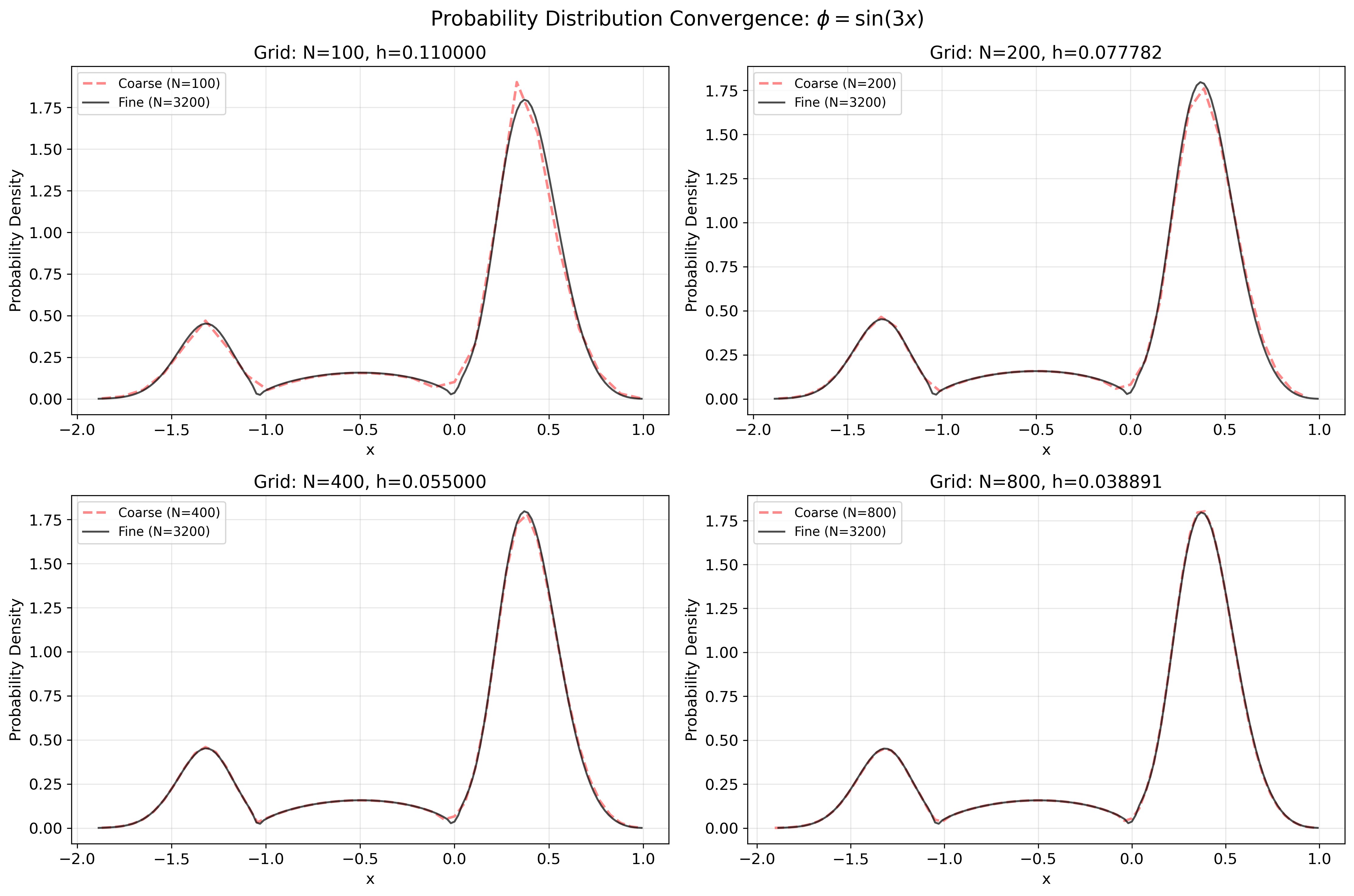}
    \caption{Convergence of probability density functions under systematic grid refinement for $\phi(x) = \sin(3x)$. The curves demonstrate that, as the discretization becomes increasingly refined, both the global structure and local features gradually approach the reference solution.}
    \label{fig:convergence_plot}
\end{figure}

\section{Conclusions}
{\label{sec:7}}

This paper introduces the notion of a responsive distribution for $G$-normal random variables and develops a coupled backward--forward trinomial tree framework for its computation.
Starting from the stochastic control representation of the $G$-expectation, we interpret the nonlinear $G$-expectation as a linear expectation under the terminal law of the optimally controlled diffusion, which naturally gives rise to the concept of a responsive distribution. 
The backward trinomial tree yields approximations of the valuen (i.e., the G-expectation) and the optimal feedback control, while the forward trinomial tree propagates the induced transition probabilities to construct a discrete approximation of the responsive distribution.
On the analytical side, we rigorously prove that the discrete value function generated by the backward trinomial tree converges to the value function of the associated stochastic optimal control problem, and consequently, to the $G$-expectation. More importantly, we establish that the discrete second-order difference quotients generated by the trinomial tree converge locally uniformly to the second spatial derivative of the backward $G$-heat equation, which directly yields the strong convergence of the discrete optimal control.
Building on this result, we further prove that the discrete probability measures generated by the forward recursion converge weakly$^\ast$ to the responsive distribution, equivalently, to a weak solution of the associated nonlinear Fokker--Planck equation. Numerical experiments support the theoretical convergence of the coupled trinomial tree schemes, indicate an approximately second-order convergence rate for the responsive distributions even under non-convex and non-concave measurements, and demonstrate that the proposed framework provides a powerful and practical sampling tool for visualizing complex, measurement-dependent responsive distributions.
 
Overall, the present work provides a unified framework for the computation and convergence analysis of responsive distributions associated with one-dimensional $G$-normal random variables under volatility uncertainty, clarifies their link with stochastic optimal control and nonlinear Fokker--Planck dynamics, and yields a practical sampling methodology relevant to financial mathematics and related problems under model uncertainty.

\bibliography{main}
\bibliographystyle{plain}


\newpage

\appendix
\section{Appendix A：Full Proofs of Convergence}

\subsection{Proof of Lemma \ref{lemma:MSC}}
\label{app_prf_lemma1}
\begin{proof}
We verify consistency, $l_\infty$-stability, and monotonicity 
under the stability condition 
$\overline{\sigma}^2 \frac{\Delta t}{h^2}\le \frac12$.

\medskip
\noindent
\textbf{(i) Consistency.}
A standard Taylor expansion shows that the forward time difference 
approximates $\partial_t$ with error $O(\Delta t)$ and the centered 
second difference approximates $\partial_{xx}$ with error $O(h^2)$. 
Hence the discrete operator converges to the G-heat operator 
as $h,\Delta t\to0$.

\medskip
\noindent
\textbf{(ii) $l_\infty$-stability.} 
Let $i_{0}$ be an index such that $\left\vert U_{i_{0}}^{n}\right\vert
=\left\vert \left\vert U^{n}\right\vert \right\vert _{\infty },$ from the
discrete scheme \eqref{dbp2} and Assumption \ref{ass:CFL_w}, we have
\begin{eqnarray}
\left\vert \left\vert U^{n}\right\vert \right\vert _{\infty } &=&\left\vert
U_{i_{0}}^{n}\right\vert  \notag \\
&\leq &\frac{\sigma _{i}^{2}}{2}\frac{\Delta t}{h^{2}}\left\vert \left\vert
U^{n+1}\right\vert \right\vert _{\infty }+\left( 1-\sigma _{i}^{2}\frac{%
\Delta t}{h^{2}}\right) \left\vert \left\vert U^{n+1}\right\vert \right\vert
_{\infty }+\frac{\sigma _{i}^{2}}{2}\frac{\Delta t}{h^{2}}\left\vert
\left\vert U^{n+1}\right\vert \right\vert _{\infty }  \notag \\
&=&\left\vert \left\vert U^{n+1}\right\vert \right\vert _{\infty }.
\end{eqnarray}%
By backward induction and $U^N=\phi$, we obtain
\[
\max_n \|U^n\|_\infty \le \|\phi\|_\infty.
\]

\medskip
\noindent
\textbf{(iii) Monotonicity.}
We consider the perturbation on $U_{i}^{n}$, it is obvious from
equation \eqref{dbp} that
\begin{equation}
g_{i}\left( U_{i}^{n}+\epsilon _{i}^{n},U_{i}^{n+1},\{U_{k}^{n+1}\}_{k\in
N_{i}}\right) -g_{i}\left( U_{i}^{n},U_{i}^{n+1},\{U_{k}^{n+1}\}_{k\in
N_{i}}\right) \leq 0.  \label{T-diag}
\end{equation}
We now turn to the perturbation on $U_{i}^{n+1}$ and $\{U_{k}^{n+1}\}_{k\in
N_{i}}.$ Denote by $\widetilde{U}_{i}^{n+1}=U_{i}^{n+1}+$ $\epsilon
_{i}^{n+1},$ for $\epsilon _{i}^{n+1}\geq 0.$ We also denote by $\widetilde{U%
}_{k}^{n+1}=U_{k}^{n+1}+\epsilon _{k}^{n+1}$, for $\epsilon _{k}^{n+1}\geq 0$
and $k\in N_{i}$. Then the difference between the two sides of the
inequality \eqref{g-off-diag} is
\begin{eqnarray}
T &:=&
g_{i}\left( U_{i}^{n},U_{i}^{n+1}+\epsilon
_{i}^{n+1},\{U_{k}^{n+1}+\epsilon _{k}^{n+1}\}_{k\in N_{i}}\right)
-g_{i}\left( U_{i}^{n},U_{i}^{n+1},\{U_{k}^{n+1}\}_{k\in N_{i}}\right)
\notag \\
&=&\frac{\epsilon _{i}^{n+1}}{\Delta t}+\frac{\widetilde{\sigma }^{2}}{2}%
\delta _{h}^{2}\widetilde{U}_{i}^{n+1}-\frac{\widehat{\sigma }^{2}}{2}\delta
_{h}^{2}U_{i}^{n+1}  \notag \\
&=&\frac{\epsilon _{i}^{n+1}}{\Delta t}+\left( \frac{\widetilde{\sigma }^{2}%
}{2}\delta _{h}^{2}\widetilde{U}_{i}^{n+1}-\frac{\widehat{\sigma }^{2}}{2}%
\delta _{h}^{2}\widetilde{U}_{i}^{n+1}\right) -\left( \frac{\widehat{\sigma }%
^{2}}{2}\delta _{h}^{2}U_{i}^{n+1}-\frac{\widehat{\sigma }^{2}}{2}\delta
_{h}^{2}\widetilde{U}_{i}^{n+1}\right)   \notag \\
&\geq &\left( \frac{1}{\Delta t}-\frac{\widehat{\sigma }^{2}}{h^{2}}\right)
\epsilon _{i}^{n+1}\label{mon1}
\end{eqnarray}%
since $\frac{\widetilde{\sigma }^{2}}{2}\delta _{h}^{2}\widetilde{U}%
_{i}^{n+1}\geq \frac{\widehat{\sigma }^{2}}{2}\delta _{h}^{2}\widetilde{U}%
_{i}^{n+1},$ where $\widetilde{\sigma }^{2}=\sigma ^{2}\left( \delta _{h}^{2}%
\widetilde{U}_{i}^{n+1}\right) ,$ $\widehat{\sigma }^{2}=\sigma ^{2}\left(
\delta _{h}^{2}U_{i}^{n+1}\right)$. Then, by Assumption \ref{ass:CFL_w}, we obtain $T\geq 0.$ The monotonicity of scheme \eqref{dbp} now follows directly from Definition \ref{def-mono}.

\medskip
This completes the proof.
\end{proof}

\subsection{Proof of monotonicity for the scheme \eqref{eq:scheme_w_def}}
\label{app:proof_mono_w}
\begin{proof}
Let $S_i^{n+1} = W_i^{n
+1} -  \widetilde W_i^{n+
1}$, we have
\[
S_i^{n}=\frac{\sigma^{2}(S_i^{n})}{\sigma^{2}(S_i^{n+1}) }
\left[ \left(1 - \frac{\Delta t\, \sigma^2(S_i^{n+1})}{h^{2}}\right)
S_i^{n+1}
+
\frac{\Delta t\, \sigma^{2}(S_i^{n+1})}{2h^{2}}
\left(
S_{i+1}^{n+1}  + S_{i-1}^{n+1}
\right)
\right].
\]
By the stability condition，$0 \le \frac{\Delta t\, \overline\sigma^{2}}{2h^{2}} \le 1$, then we have
\[
0 \le \frac{\Delta t\, \sigma^{2}(S_i^{n+1})}{2h^{2}} \le 1,
\quad
0 \le 1 - \frac{\Delta t\, \sigma^2(S_i^{n+1})}{h^{2}} \le 1, \quad
\frac{\sigma^{2}(S_i^{n})}{\sigma^{2}(S_i^{n+1})} > 0.
\]
Therefore, for all \(\Delta t\) sufficiently small,
\[
S_i^{n} \ge0, \quad if \;\; S_i^{n+1} \ge 0.
\]
\end{proof}

\subsection{Proof of viscosity consistency \eqref{eq:cons_w}}
\label{app:proof_cons_w}

\begin{proof}
We prove the consistency in the viscosity sense for the discrete operator
$L^{h,\Delta t}$ defined in \eqref{eq:scheme_w_def}. Since the arguments for
$\limsup$ and $\liminf$ are symmetric (with $F^*$ and $F_*$), we present the
$\limsup$ part; the $\liminf$ part follows analogously.

As in the standard viscosity framework for discontinuous nonlinearities, define the upper semicontinuous envelope of $F$,
\[
\limsup_{(\tau,y)\to (t_{0},x_{0}) \atop (\tau,y)\in (0,T)\times\Omega}
\partial_t \psi (\tau,y) + \tfrac12 {\sigma}^{2}(\psi)\partial_{xx}\psi(\tau,y)
=
\begin{cases}
\partial_t \psi (t_0,x_0)+ \tfrac12\bar{\sigma}^{2} \partial_{xx}\psi(t_0,x_0), & \psi>0,\\[4pt]
\partial_t \psi (t_0,x_0) + \tfrac12 \max\{\overline{\sigma}^{2}\partial_{xx}\psi,\; \underline{\sigma}^{2}\partial_{xx}\psi\}(t_0,x_0), & \psi=0,\\[4pt]
\partial_t \psi (t_0,x_0) + \tfrac12 \underline{\sigma}^{2}\partial_{xx}\psi(t_0,x_0), & \psi<0.
\end{cases}
\]

Let $\psi^{h,\Delta t}$ be the piecewise-constant interpolation of $\psi$ on the grid,
the scheme uses the operator:
\begin{equation}\label{eq:scheme_consist}
\small
\begin{aligned}
&\mathcal{L}^{h,\Delta t}(t,x,\psi^{h,\Delta t}(t,x),\psi) \\
= &
\Bigl[
\psi^{h,\Delta t}_{t+\Delta t, x}
-
\frac{\sigma^{2}(\psi^{h,\Delta t}_{t+\Delta t,x})}
     {\sigma^{2}(\psi^{h,\Delta t}_{t,x})}
\psi^{h,\Delta t}_{t,x}
\Bigr]\!/\Delta t + 
\frac{\sigma^{2}(\psi^{h,\Delta t}_{t+\Delta t,x})}{2}
\frac{
\psi^{h,\Delta t}_{t+\Delta t,x+h}
-2\psi^{h,\Delta t}_{t+\Delta t,x}
+ \psi^{h,\Delta t}_{t+\Delta t,x-h}
}{h^{2}}\\
= &\frac{
\psi^{h,\Delta t}_{t+\Delta t,x}
-
\psi^{h,\Delta t}_{t,x}
}{\Delta t}
+
\frac{
\psi^{h,\Delta t}_{t,x}
-
\frac{\sigma^{2}(\psi^{h,\Delta t}_{t+\Delta t,x})}{\sigma^{2}(\psi^{h,\Delta t}_{t,x})}\psi^{h,\Delta t}_{t,x}
}{\Delta t}
+
\frac{
\psi^{h,\Delta t}_{t+\Delta t,x+h}
- 2\psi^{h,\Delta t}_{t+\Delta t,x}
+ \psi^{h,\Delta t}_{t+\Delta t,x-h}
}{h^{2}}
\frac{\sigma^{2}(\psi^{h,\Delta t}_{t,x})}{2}.
\end{aligned}
\end{equation}
and 
\begin{equation}\label{eq:app_.6}
\begin{aligned}
\limsup_{h,\Delta t\to 0,\xi \to0, \atop (\tau,y)\to (t_{0},x_{0})}
L^{h,\Delta t}(\tau,y,\psi^{h,\Delta t}_{\tau,y}+\xi, \psi^{h,\Delta t}+\xi)
= \limsup_{h,\Delta t\to 0,\xi \to0, \atop (\tau,y)\to (t_{0},x_{0})} A^{h,\Delta t}+B^{h,\Delta t}+C^{h,\Delta t},
\end{aligned}
\end{equation}
with
\begin{align*}
A_{h,\Delta t}
&:=\frac{
\psi^{h,\Delta t}_{\tau+\Delta t,y}-\psi^{h,\Delta t}_{\tau,y}}{\Delta t},\\[2pt]
B_{h,\Delta t}
&:=\frac{
\psi^{h,\Delta t}_{\tau,y}+\xi
-
\frac{\sigma^{2}(\psi^{h,\Delta t}_{\tau+\Delta t,y})}
     {\sigma^{2}(\psi^{h,\Delta t}_{\tau,y})}
(\psi^{h,\Delta t}_{\tau,y}+\xi)
}{\Delta t},\\[2pt]
C_{h,\Delta t}
&:=\frac{\sigma^{2}(\psi^{h,\Delta t}_{\tau,y})}{2}
\frac{
\psi^{h,\Delta t}_{\tau+\Delta t, y+h}
-2\psi^{h,\Delta t}_{\tau+\Delta t,y}
+ \psi^{h,\Delta t}_{\tau+\Delta t, y-h}
}{h^{2}}.
\end{align*}

\medskip
\noindent\textbf{Step 1: time term.}
By standard consistency of the forward difference quotient for smooth $\psi$,
\begin{equation}
\begin{aligned}
\limsup_{h,\Delta t\to 0,\xi \to0, \atop (\tau,y)\to (t_{0},x_{0})}
\frac{
\psi^{h,\Delta t}_{\tau+\Delta t,y}
-
\psi^{h,\Delta t}_{\tau,y}
}{\Delta t} = \partial_t \psi(t_0,x_0)
\end{aligned}
\end{equation}

\medskip
\noindent\textbf{Step 2: diffusion term.} 
When $\psi^{h,\Delta t}(t_0,x_0)\neq 0$, we have
\[
\limsup_{h,\Delta t\to 0,\ \xi \to 0,\ (\tau,y)\to (t_{0},x_{0})}
\Bigg(
1-\frac{\sigma^{2}\big(\psi^{h,\Delta t}_{\tau+\Delta t,y}\big)}
           {\sigma^{2}\big(\psi^{h,\Delta t}_{\tau,y}\big)}
\Bigg)=0,
\]
and therefore
\begin{equation}
\begin{aligned}
\limsup_{h,\Delta t\to 0,\xi \to0, \atop (\tau,y)\to (t_{0},x_{0})}
B_{h,\Delta t}
&= \limsup_{h,\Delta t\to 0,\xi \to0, \atop (\tau,y)\to (t_{0},x_{0})} \frac{
1
-
\frac{\sigma^{2}(\psi^{h,\Delta t}_{\tau+\Delta t,y})}
     {\sigma^{2}(\psi^{h,\Delta t}_{\tau,y})}
}{\Delta t} \cdot \limsup_{h,\Delta t\to 0 \atop (\tau,y)\to (t_{0},x_{0})}(\psi^{h,\Delta t}_{\tau,y}+\xi)=0,
\end{aligned}
\end{equation}
\begin{align}
\limsup_{h,\Delta t\to 0,\xi \to0, \atop (\tau,y)\to (t_{0},x_{0})}
C_{h,\Delta t}
&=
\begin{cases}
 \tfrac12\bar{\sigma}^{2} \partial_{xx}\psi(t_0,x_0), & \psi^{h,\Delta t}_{t_0, x_0}>0,\\[4pt]
  \tfrac12 \max\{\overline{\sigma}^{2}\partial_{xx}\psi,\; \underline{\sigma}^{2}\partial_{xx}\psi\}(t_0,x_0), & \psi^{h,\Delta t}_{t_0, x_0}=0,\\[4pt]
 \tfrac12 \underline{\sigma}^{2}\partial_{xx}\psi(t_0,x_0), & \psi^{h,\Delta t}_{t_0, x_0}<0.
\end{cases}
\end{align}
Combining Steps 1--2, equation \eqref{eq:app_.6} can be rewritten as:
\begin{equation}\label{eq:scheme_consist}
\begin{aligned}
&\limsup_{h,\Delta t\to 0,\xi \to0, \atop (\tau,y)\to (t_{0},x_{0})}
L^{h,\Delta t}(\tau,y,\psi^{h,\Delta t}_{\tau,y}+\xi, \psi^{h,\Delta t}+\xi)\\
=&\begin{cases}
 \partial_t \psi (t_0,x_0)+ \tfrac12\bar{\sigma}^{2} \partial_{xx}\psi(t_0,x_0), & \psi^{h,\Delta t}_{t_0, x_0}>0,\\[4pt]
  \partial_t \psi (t_0,x_0)+ \tfrac12 \max\{\overline{\sigma}^{2}\partial_{xx}\psi,\; \underline{\sigma}^{2}\partial_{xx}\psi\}(t_0,x_0), & \psi^{h,\Delta t}_{t_0, x_0}=0,\\[4pt]
 \partial_t \psi (t_0,x_0)+ \tfrac12 \underline{\sigma}^{2}\partial_{xx}\psi(t_0,x_0), & \psi^{h,\Delta t}_{t_0, x_0}<0,
\end{cases}\\
=& \limsup_{(\tau,y) \to (t_0,x_0), \atop(\tau,y)\in[0,T)\times\Omega} \partial_t\psi(\tau,y) + \frac{1}{2}\sigma^2(\psi)\partial_{xx}\left(  \psi\right)(\tau,y),
\end{aligned}
\end{equation}
which is exactly the required viscosity sub-consistency. The super-consistency (the $\liminf$ inequality)
is proved in the same way, using the lower semicontinuous envelope $F_\ast$.

\end{proof}

\subsection{Proof of Theorem~\ref{thm:conv_w}}
\label{app:proof_conv_w}

\begin{proof}

Let \(\overline{w},\,\underline{w} \in B((0,T)\times\bar{\Omega})\) be defined by
\[
\overline{w}
   = \limsup_{h,\Delta t\to 0} w^{h,\Delta t},
\qquad
\underline{w}
   = \liminf_{h,\Delta t\to 0} w^{h,\Delta t}.
\]
By construction, $\overline{w}$ is upper semicontinuous and
$\underline{w}$ is lower semicontinuous, and clearly
$\underline{w}\le\overline{w}$ on $(0,T)\times\overline\Omega$.
We show that $\overline{w}$ is a viscosity subsolution of
\eqref{eq:w-equation}; the proof that $\underline{w}$ is a viscosity
supersolution is entirely analogous.

Let $(t_0,x_0)\in(0,T)\times\Omega$ and let
$\psi\in C^\infty_b((0,T)\times\overline\Omega)$ be such that
$\overline{w}-\psi$ attains a local maximum at $(t_0,x_0)$. Without loss of
generality, we may assume that this maximum is strict and that
\[
\overline{w}(t_0,x_0)=\psi(t_0,x_0).
\]
Moreover, by modifying $\psi$ outside a small neighborhood of
$(t_0,x_0)$, we may assume that
\[
\psi(t,x)\ge 2\sup_{h,\Delta t}\|w^{h,\Delta t}\|_\infty
\quad\text{for }(t,x)\notin B_r(t_0,x_0),
\]
for some $r>0$, and that
\[
\overline{w}(t,x) - \psi(t,x) \le0= \overline{w}(t_{0},x_{0}) - \psi(t_{0},x_{0}) \quad in \; B((t_0,x_0), r).
\]
By the definition of $\overline{w}$, there exist sequences
$\rho_n=(h_n,\Delta t_n)\to0$ and $(t_n,y_n)\to(t_0,x_0)$ such that
\[
w^{\rho_n}(t_{n},y_{n}) \to \overline{w}(t_{0},x_{0}),
\qquad n\to\infty,
\]
where each \((t_{n},y_{n})\) is a global maximum point of  
$w^{\rho_n}-\psi$. Define
\[
\xi_n:=w^{\rho_n}(t_n,y_n)-\psi(t_n,y_n)\to0.
\]
Then
\[
w^{\rho_n}(t,x)\le \psi(t,x)+\xi_n
\qquad\text{for all }(t,x)\in(0,T)\times\overline\Omega.
\]

By the definition of \(w^{\rho} = w^{h,\Delta t}\),  
and by the monotonicity of the scheme \(L^{\rho} = L^{h,\Delta t}\),  
together with
$w^{\rho_n}\le \psi + \xi_n,$
we obtain
\begin{align}
0=&L^{\rho_n}\bigl( \rho_n, \tau_n,y_n,\, w^{\rho_n}(t_n,y_n),\,
              w^{\rho_n} \bigr) \\
\le&L^{\rho_n}\bigl( \rho_n,\tau_n,y_n,\, \psi(t_n,y_n)+\xi_n,\,
              \psi+\xi_n \bigr).     
\end{align}
Using the consistency of the scheme of \eqref{eq:scheme_consist}, above yield
\begin{equation}
\begin{aligned}
0
\;\le\;
&\liminf_{\rho_n\to\infty}
L^{\rho_n}\bigl((\tau_n,y_n),\, \psi(\tau_n,y_n)+\xi_n,\,
\psi + \xi_n \bigr)\\
\;\le\;& \liminf_{\substack{(\tau,y)\to (t_0,x_0),\atop n\to\infty,\xi\to 0}}
L\bigl( (\tau,y),\, \psi(\tau,y)+\xi,\, \psi+\xi \bigr)\\
\;\le\;& \partial_t \psi(t_0，x_0) + F_{*}(D^2 \psi, \psi, t_0, x_0).
\end{aligned}
\end{equation}
Since $\overline{w}(t_0,x_0)=\psi(t_0,x_0)$, we have
\[
\partial_t \psi(t_0，x_0) + F_{*}(D^2 \psi, \overline{w}, t_0, x_0)
\]
As $(t_0,x_0)$ was
arbitrary, $\overline{w}$ is a viscosity subsolution of
\eqref{eq:w-equation},
\[
\partial_t \psi(t_0，x_0) + F_{*}(D^2 \psi, \overline{w}, t, x) \ge 0.
\]
This shows that \(\overline{w} = \limsup_{h,\Delta t\to 0} w^{h,\Delta t}\) is a viscosity subsolution of equation \eqref{eq:w-equation}.
Similarly,  
$\underline{w}
=
\liminf_{h,\Delta t\to 0} w^{h,\Delta t}$
is a viscosity supersolution of equation \eqref{eq:w-equation}.
\end{proof}

\subsection{Proof of Remark~\ref{thm:conv_v}}
\label{app:proof_conv_v}
Define the function \( F(v) = \sigma^2(v)\,v \), where the function \( \sigma \) satisfies the condition specified in Equation~\eqref{eq:sigma_optimal}. Then,
\[
F(v) =
\begin{cases}
\overline{\sigma}^2 v, & v \ge 0, \\
\underline{\sigma}^2 v, & v < 0.
\end{cases}
\]
It is straightforward to verify that, for any \( x, y \in \mathbb{R} \),
\[
|F(x) - F(y)| \ge \underline{\sigma}^2 |x - y|.
\]
Hence, we have the pointwise inequality
\[
|V_i^n - u_{xx}(t_n, x_i)| \le \frac{1}{\underline{\sigma}^2} \, |F(V_i^n) - F(u_{xx}(t_n, x_i))|.
\]
Taking the \( L^p \)-norm on both sides yields
\[
\|V - u_{xx}\|_{L^p} \le \frac{1}{\underline{\sigma}^2} \, \|F(V) - F(u_{xx})\|_{L^p}.
\]
Since the right-hand side \( \|F(V_i^n) - F(u_{xx})\|_{L^p} \) converges uniformly to zero, it follows that
\[
V_i^n \to u_{xx}(t_n, x_i), \quad \text{uniformly.} \quad \square
\]

\subsection{Proof of Lemma~\ref{lem:probability_preserving}}
\label{app:proof_probability_preserving}
\begin{proof}
The trinomial forward step can be written as
\begin{align}\label{eq:fwd}
		p_i^n &= P_{i-1, i}^n (\sigma_{i-1}^{n-1}) \ast p_{i-1}^{n-1} + P_{i,i}^n (\sigma_i^{n-1}) \ast p_i^{n-1} + P_{i+1, i}^n (\sigma_{i+1}^{n-1}) \ast p_{i+1}^{n-1},\ i = -n, \dots, n,
\end{align}
where for any $\sigma\in[\underline{\sigma},\overline{\sigma}]$,
\[
P_{i,i-1}^n(\sigma) = P_{i,i+1}^n(\sigma) = \frac{\sigma^2}{2} \frac{\Delta t}{h^2}, \quad P_{i,i}^n(\sigma) = 1 - P_{i,i-1}^n(\sigma) - P_{i,i+1}^n(\sigma).
\]
Under \eqref{ass:CFL_w}, all transition weights are nonnegative and satisfy
$$P_{i,i-1}^n(\sigma)+P_{i,i}^n(\sigma)+P_{i,i+1}^n(\sigma)=1.$$

(1) Positivity follows by induction from $p_i^0\ge 0$ and the fact that $p_i^n$ is a
nonnegative linear combination of $\{p_{i-1}^{n-1},p_i^{n-1},p_{i+1}^{n-1}\}$.

(2) Mass conservation: summing \eqref{eq:fwd} over $i$ and using the stochasticity of the
weights yields $\sum_i p_i^n = \sum_i p_i^{n-1}$, hence $\sum_i p_i^n  = \sum_i p_i^0  = 1$.
\end{proof}

\subsection{Proof of Lemma~\ref{lem:sig_V_to_sig_uxx}}
\label{app:vitali_sigma}

\begin{proof}
\smallskip
\noindent\textbf{Step 1: Almost everywhere convergence.}
By the Theorem \ref{thm:conv_v}, we have
\[
V_i^n \;\longrightarrow\; u_{xx}(t_n,x_i)
\quad \text{for a.e. } (t,x)\in [0,1]\times \overline \Omega, \qquad \text{as } h\to 0.
\]
Since $\sigma^2(\cdot)$ has only finitely many points of discontinuity, the limit
$u_{xx}(t_n,x_i)$ avoids these discontinuities for almost every $x\in\Omega$.
At such continuity points, composition preserves convergence, and therefore
\[
\sigma^2(V_i^n) \;\longrightarrow\; \sigma^2\!\big(u_{xx}(t_n,x_i)\big)
\quad \text{for a.e. } x\in\Omega .
\]

\smallskip
\noindent\textbf{Step 2: Uniform integrability.}
The function $\sigma^2(\cdot)$ is bounded, i.e.,
\(
0 \le \sigma^2(\cdot) \le \overline{\sigma}^{\,2}.
\)
Hence, for any $\varepsilon>0$, choosing
\(
\delta=\varepsilon/\overline{\sigma}^{\,2},
\)
we have, for any measurable set $A\subset\Omega$ with $\mu(A)<\delta$,
\[
\int_A \big|\sigma^2(V_i^n)\big|\,d\mu
\;\le\;
\overline{\sigma}^{\,2}\,\mu(A)
\;\le\;
\varepsilon .
\]
This shows that the family $\{\sigma^2(V_i^n)\}$ is uniformly integrable.

\smallskip
\noindent\textbf{Step 3: Strong convergence.}
Combining almost everywhere convergence with uniform integrability, Vitali's convergence
theorem yields
\[
\sigma^2(V_i^n)
\;\longrightarrow\;
\sigma^2\!\big(u_{xx}(t_n,x_i)\big)
\quad \text{strongly in } L^p([0,T] \times \Omega),
\qquad \forall\,1\le p<\infty .
\]
In particular,
\[
\int_\Omega
\big|\sigma^2(V_i^n)-\sigma^2\!\big(u_{xx}(t_n,x_i)\big)\big|^{p}\,dx
\;\longrightarrow\; 0,
\qquad \text{as } h\to 0,
\]
which completes the proof.
\end{proof}

\subsection{Proof of Theorem~\ref{thm:weakstar_FP}}
\label{app:proof_thm_weakstar_FP}

\begin{proof}
Fix $T\in(0,1)$ and let $N$ be such that $t_N=N\Delta t=T$.
The proof is divided into three steps.

\smallskip
\noindent\textbf{Step 1. Discrete weak formulation.}
Let $\varphi\in C_c^{\infty}([0,T]\times\mathbb R)$ with $\varphi(\cdot,T)=0$, and denote
$\varphi_i^n:=\varphi(x_i,t_n)$.
Multiply the scheme \eqref{eq:p_fdm} by $\varphi_i^{n+1}h\Delta t$ and sum over
$i\in\mathbb Z$ and $n=0,\dots,N-1$ to obtain
\begin{equation}\label{eq:sum_app}
\sum_{n=0}^{N-1}\sum_{i=-\infty}^{\infty}
h\Delta t\left[
-\,\frac{p_i^{\,n+1}-p_i^{\,n}}{\Delta t}\,\varphi_i^{n+1}
+\frac12\,\varphi_i^{n+1}\,
\frac{ (\sigma_{i+1}^{\,n})^2 p_{i+1}^{\,n}
-2(\sigma_i^{\,n})^2 p_i^{\,n}
+(\sigma_{i-1}^{\,n})^2 p_{i-1}^{\,n}}{h^2}
\right]=0,
\end{equation}
where $\sigma_i^n=\sigma(V_i^n)$.

\medskip
\noindent
\emph{(i) Time summation-by-parts.}
Using the discrete identity
\[
-\sum_{n=0}^{N-1}
(p_i^{n+1}-p_i^n)\varphi_i^{n+1}
=
\sum_{n=0}^{N-1}
p_i^n(\varphi_i^{n+1}-\varphi_i^n)
-p_i^N\varphi_i^N+p_i^0\varphi_i^0,
\]
we rewrite the time contribution in 
\begin{equation}\label{eq:proof_8.1}
\begin{aligned}
&\sum_{n=0}^{N-1} \sum_{i=-\infty}^\infty -\frac{p_i^{n+1} - p_i^n}{\Delta t} \varphi_i^{n+1} h \Delta t\\
=&\sum_{n=0}^{N-1} \sum_{i=-\infty}^\infty  p_i^n\frac{\varphi_i^{n+1} - \varphi_i^n}{\Delta t} h \Delta t 
+ \sum_{i=-\infty}^\infty \left[  p_i^0 \varphi_i^0 - p_i^N \varphi(T,x_i)\right]h.
\end{aligned}
\end{equation}

\medskip
\noindent
\emph{(ii) Spatial summation-by-parts.}
For the diffusion term in \eqref{eq:sum_app}, using the compact support of $\varphi$ in $x$ and a
discrete integration-by-parts identity, we obtain
\begin{equation}\label{eq:space_part_app}
\begin{aligned}
&\sum_{n=0}^{N-1} \sum_{i=-\infty}^\infty \frac{h \Delta t}{2} \frac{(\sigma_{i+1}^n)^2 p_{i+1}^n - 2 (\sigma_i^n)^2 p_i^n + (\sigma_{i-1}^n)^2 p_{i-1}^n}{h^2} \varphi_i^{n+1}\\
=&\sum_{n=0}^{N-1} \sum_{i=-\infty}^\infty \frac{h \Delta t}{2} (\sigma_{i}^n)^2 p_i^n \frac{(\varphi_{i+1}^n - 2 \varphi_{i}^n + \varphi_{i-1}^n)}{h^2}\\
=&\sum_{n=0}^{N-1} \sum_{i=-\infty}^\infty \frac{h \Delta t}{2} (\sigma_{i}^n)^2 p^{h,\Delta t} \frac{(\varphi_{i+1}^n - 2 \varphi_{i}^n + \varphi_{i-1}^n)}{h^2}.
\end{aligned}
\end{equation}
Then the equation \eqref{eq:sum_app} can be rewritten as
\begin{equation}\label{eq_prf8_1}
\sum_{i=-\infty}^\infty \left[  p_i^0 \varphi_i^0 - p_i^N \varphi(T,x_i) \right] h + 
\sum_{n=0}^{N-1} \sum_{i=-\infty}^\infty 
p_i^n \left[ \frac{\varphi_{i}^{n+1} - \varphi_{i}^{n}}{\Delta t} + \frac{\left( \sigma_i^n \right)^2}{2} \frac{(\varphi_{i+1}^n - 2 \varphi_{i}^n + \varphi_{i-1}^n)}{h^2} \right]h \Delta t.
\end{equation}

\medskip
\noindent
\textbf{Step 2. Decomposition into integral form and remainder.} We define $u^{h,\Delta t}(x,y,t)$,
$v_i^{h,\Delta t}(s,t)$,
$w_i^{h,\Delta t}(t)$, $i=1,2,3$, as piecewise constant functions:
\begin{equation}\label{2.12}
\left\{
\begin{aligned}
p^{h,\Delta t}(t,x)
&= P_i^n,
&& (t,x) \in (x_i-\tfrac{h}{2},\,x_i+\tfrac{h}{2}] \times \bigl((n-1)\Delta t,\, n\Delta t\bigr], \\[4pt]
u^{h,\Delta t}(t,x)
&= U_i^n,
&& (t,x) \in (x_i-\tfrac{h}{2},\,x_i+\tfrac{h}{2}] \times \bigl((n-1)\Delta t,\, n\Delta t\bigr], \\[4pt]
v^{h,\Delta t}(t,x)
&= V_i^n,
&& (t,x) \in (x_i-\tfrac{h}{2},\,x_i+\tfrac{h}{2}] \times \bigl((n-1)\Delta t,\, n\Delta t\bigr], \\[4pt]
w^{h,\Delta t}(t,x)
&= W_i^n,
&& (t,x) \in (x_i-\tfrac{h}{2},\,x_i+\tfrac{h}{2}] \times \bigl((n-1)\Delta t,\, n\Delta t\bigr].
\end{aligned}
\right.
\end{equation}
Denote
$B_T := \overline{B} \times [0,T]$.
For a temporal function $\psi \in C^0[0,T]$, we define its piecewise
interpolation $I_T^{\Delta t}\psi$ in a backward way as in
Definition~\eqref{2.12}.
For a spatial function $\psi \in C^0(\overline{\Omega})$, we define its
piecewise interpolations $I_B^{h}\psi$ in a central way as in
Definition~\eqref{2.12} with $B=\Omega$.
For any $\psi \in C^0(\Omega_T)$, if we define its piecewise interpolation
\[
I_{B_T}\psi = I_T^{\Delta t} \circ I_B^{h}\psi,
\quad \text{with } B=\Omega,
\]
then we have
\begin{equation}
\label{2.13}
\|\psi - I_{B_T}\psi\|_{L^\infty(B_T)}
\longrightarrow 0,
\quad \text{as } h,\Delta t \to 0,
\text{ for } B=\Omega, \Gamma_i, \text{ or } P_i.
\end{equation}

\begin{lemma} 
For $\varphi \in C^{\infty}(\Omega_T)$, if we denote by $\square^{h, \Delta t}\varphi_i^n = \frac{\left( \sigma_i^n \right)^2}{2} \frac{(\varphi_{i+1}^n - 2 \varphi_{i}^n + \varphi_{i-1}^n)}{h^2}$, and by $\square \varphi = \partial_t \varphi + \mathcal{L} ^* \varphi = \partial_t \varphi + \frac{\left( \sigma\left( \partial_{xx} u\right) \right)^2}{2} \partial_{xx} \varphi$, we have
 $E=\square^{h, \Delta t}\varphi_i^n - \square \varphi \to 0$, when $h, \Delta t  \to 0$.
\end{lemma}

\begin{proof}
For the time derivative $\phi_t$, a first–order Taylor expansion yields
\[
\varphi(t_{n+1}, x_i)
= \varphi(t_{n}, x_i)
+ \partial_t \varphi(t_n, x_i)\Delta t
+ O(\Delta t^2).
\]
Hence,
\[
\frac{\varphi(t_{n+1}, x_i) - \varphi(t_{n}, x_i)}{\Delta t}
= \partial_t \varphi(t_n, x_i)
+ O(\Delta t).
\]
For the second–order spatial derivative (taking the $x$–direction as an example),
the Taylor expansions read
\[
\varphi(t_{n}, x_{i+1})
= \varphi(t_{n}, x_{i})
+ h \partial_x \varphi(t_{n}, x_{i})
+ \frac{h^2}{2} \partial_x^2 \varphi(t_{n}, x_{i})
+ \frac{h^3}{6} \partial_x^3 \varphi(t_{n}, x_{i})
+ O(h^4),
\]
\[
\varphi(t_{n}, x_{i-1})
= \varphi(t_{n}, x_{i})
- h \partial_x \varphi(t_{n}, x_{i})
+ \frac{h^2}{2} \partial_x^2 \varphi(t_{n}, x_{i})
- \frac{h^3}{6} \partial_x^3 \varphi(t_{n}, x_{i})
+ O(h^4).
\]
Adding them together, we get:
\[
\frac{\varphi(t_{n}, x_{i+1}) - 2 \varphi(t_{n}, x_{i}) + \varphi(t_{n}, x_{i-1})}{h^2} = \partial_x^2 \varphi(t_{n}, x_{i}) + O(h^2).
\]
Thus, 
\begin{equation}\label{eq:lemma.1}
\begin{aligned}
\max_{\substack{x_i\in \Omega_n, t_n \in [0,T]}}
\left| (\square \varphi)_{i}^n - (\square^{h, \Delta t}\varphi)_{i}^n \right| &= O(\Delta t) + \frac{\left( \sigma_i^n \right)^2}{2} O(h^2) + \frac{1}{2} \left[ \left( \sigma_i^n \right)^2 - \sigma^2(\partial_x^2 u) \right] O(h^2)  \to 0.
\end{aligned}
\end{equation}
\end{proof}

For the summation of interior points,
\[
\begin{aligned}
&\sum_{n=0}^{N} \sum_{i=-\infty}^\infty 
p_i^n \left[ \frac{\varphi_{i}^{n+1} - \varphi_{i}^{n}}{\Delta t} + \frac{\left( \sigma_i^n \right)^2}{2} \frac{(\varphi_{i+1}^n - 2 \varphi_{i}^n + \varphi_{i-1}^n)}{h^2} \right]h \Delta t \\
&= \int_0^T \int_\Omega   p^{h,\Delta t}(t,x) I_{\Omega_T}(\square\varphi)(t,x) \, dx dt +  \sum_{n=0}^{N} \sum_{i=-\infty}^\infty   p^{h,\Delta t}(t,x) (\square^{h,\Delta t}\varphi_i^n - \square\varphi_i^n) \, h \Delta t.
\end{aligned}
\]
With all of these notations, \eqref{eq_prf8_1} can be rewritten as, for $\varphi\in C_c^{\infty}(\Omega_T)$ 
\begin{equation}\label{eq.21}
\begin{aligned}
0 = &\int_\Omega \delta(x-0) \varphi(0,x) dx - \int_\Omega p^{h,\Delta t}(T,x) I_{\Omega}^h \varphi(T+\Delta t)dx \\
&+ \int_0^T \int_\Omega   p^{h,\Delta t}(t,x) I_{\Omega_T}(\square\varphi)(t,x) \, dx dt +  \sum_{n=0}^{N} \sum_{i=-\infty}^\infty   p^{h,\Delta t}(t,x) (\square^{h,\Delta t}\varphi_i^n - \square\varphi_i^n) \, h \Delta t.
\end{aligned}
\end{equation}

\medskip
\noindent\textbf{Step 3. Passage to the limit.}
By Lemma~\ref{lem:sig_V_to_sig_uxx} and \ref{lemma:Radon}, we have, for any fixed $T > 0$
and up to a subsequence, that, when $h, \Delta t  \rightarrow 0$,
\[
\begin{aligned}
&p^{h,\Delta t} \to p(t,x) \in M([0,T)\times \mathbb R) \quad weakly^* \; in \; M([0,T)\times \mathbb R),\\
&p^{h,T} \to p(T,x) \in M(\mathbb R) \quad weakly^* \; in \; M(\mathbb R),
\end{aligned}
\]
and
\[
\sigma^2(V_i^n)\ \longrightarrow\ \sigma^2(u_{xx}(t_n,x_i))
\quad\text{strongly in }L^p((0,T)\times {\mathbb R})\ \text{for any }1\le p<\infty.
\]
Taking the limit of \eqref{eq.21} along the above subsequence,
and using \eqref{2.13} together with \eqref{eq:lemma.1},
we obtain that, for any $\varphi \in C_c^{\infty}(\Omega_T)$,
as $h,\Delta t \to 0$,
\begin{equation}\label{eq:limit_identity}
\begin{aligned}
0
=& \int_\Omega \delta(x-0) \varphi(0,x) dx
   - \int_{\Omega} p(T,x)\varphi(T,x)\,dx  \\
&+ \int_0^T \int_{\Omega}
   p(t,x)\Big(
      \partial_t \varphi(t,x)
      + \frac{\sigma^2(u_{xx}(t,x))}{2}
        \partial_{xx}\varphi(t,x)
   \Big)\,dx\,dt .
\end{aligned}
\end{equation}
Here the convergence of the terminal term follows from
\[
p^{h,\Delta t}(T,\cdot)\stackrel{*}{\rightharpoonup} p(T,\cdot)
\quad\text{in } M(\mathbb R),
\]

\medskip

On the other hand, by the definition of weak solution
(see Definition~\ref{eq:weak_form}),
for any $\xi \in C_c^{\infty}(\Omega_T)$ the weak formulation reads
\begin{equation}\label{eq:weak_form_def}
\begin{aligned}
0
=& \langle \delta(x-x_0), \xi(0,x) \rangle
   - \int_{\Omega} p(T,x)\xi(T,x)\,dx \\
&+ \int_0^T \int_{\Omega}
   p(t,x)\Big(
      \partial_t \xi(t,x)
      + \frac{\sigma^2(u_{xx}(t,x))}{2}
        \partial_{xx}\xi(t,x)
   \Big)\,dx\,dt .
\end{aligned}
\end{equation}
Comparing \eqref{eq:limit_identity} and \eqref{eq:weak_form_def},
we conclude that the limit measure $p(t,x)$ satisfies
the weak formulation of the Fokker–Planck equation with initial datum
$\delta_{x_0}$.
Therefore $p$ is a weak solution of problem \eqref{eq:p-Fokker-plank}.

\medskip

Finally, by the {uniqueness} of weak solutions to
\eqref{eq:p-Fokker-plank},
the whole sequence $p^{h,\Delta t}$ converges
weakly-$^*$ towards $p$ in
$M([0,T]\times\mathbb R)$.
Since $T>0$ is arbitrary, the convergence holds on every finite time interval.
This completes the proof of Theorem~\ref{thm:weakstar_FP}.


\end{proof}

\end{document}